\PassOptionsToPackage{unicode}{hyperref}
\PassOptionsToPackage{hyphens}{url}
\PassOptionsToPackage{dvipsnames,svgnames,x11names}{xcolor}
\documentclass[
  12pt,
]{article}
\usepackage{xcolor}
\usepackage[margin=1.25in]{geometry}
\usepackage{amsmath,amssymb}
\usepackage{iftex}
\ifPDFTeX
  \usepackage[T1]{fontenc}
  \usepackage[utf8]{inputenc}
  \usepackage{textcomp} 
\else 
  \usepackage{unicode-math} 
  \defaultfontfeatures{Scale=MatchLowercase}
  \defaultfontfeatures[\rmfamily]{Ligatures=TeX,Scale=1}
\fi
\usepackage{lmodern}
\ifPDFTeX\else
\fi
\IfFileExists{upquote.sty}{\usepackage{upquote}}{}
\IfFileExists{microtype.sty}{
  \usepackage[]{microtype}
  \UseMicrotypeSet[protrusion]{basicmath} 
}{}
\makeatletter
\@ifundefined{KOMAClassName}{
  \IfFileExists{parskip.sty}{%
    \usepackage{parskip}
  }{
    \setlength{\parindent}{0pt}
    \setlength{\parskip}{6pt plus 2pt minus 1pt}}
}{
  \KOMAoptions{parskip=half}}
\makeatother

\usepackage{longtable,booktabs,array}
\usepackage{calc} 
\usepackage{etoolbox}
\makeatletter
\patchcmd\longtable{\par}{\if@noskipsec\mbox{}\fi\par}{}{}
\makeatother
\IfFileExists{footnotehyper.sty}{\usepackage{footnotehyper}}{\usepackage{footnote}}
\makesavenoteenv{longtable}
\usepackage{graphicx}
\makeatletter
\newsavebox\pandoc@box
\newcommand*\pandocbounded[1]{
  \sbox\pandoc@box{#1}%
  \Gscale@div\@tempa{\textheight}{\dimexpr\ht\pandoc@box+\dp\pandoc@box\relax}%
  \Gscale@div\@tempb{\linewidth}{\wd\pandoc@box}%
  \ifdim\@tempb\p@<\@tempa\p@\let\@tempa\@tempb\fi
  \ifdim\@tempa\p@<\p@\scalebox{\@tempa}{\usebox\pandoc@box}%
  \else\usebox{\pandoc@box}%
  \fi%
}
\def\fps@figure{htbp}
\makeatother

\providecommand{\tightlist}{%
  \setlength{\itemsep}{0pt}\setlength{\parskip}{0pt}}

\usepackage[style=chicago-authordate,backend=biber,pagetracker=true,autocite=inline,alldates=comp,labeldateparts=true,citetracker=true,uniquename=minfull,useeditor=true,usetranslator=true,usenamec=true,alltimes=12h,urltime=24h,datecirca=true,datezeros=false,dateuncertain=true,timezones=true,compressyears=true,ibidtracker=constrict,sorting=cms,punctfont,cmslos=true,nodates,uniquelist=minyear,maxbibnames=10,minbibnames=7,sortcase=false,abbreviate=false,dateabbrev=false,avdate=true]{biblatex}
\usepackage{amsthm}

\theoremstyle{newstyle}

\usepackage[ruled]{algorithm2e}
\usepackage{xspace}
\usepackage{siunitx}
\newcommand{\countrylabel}{A\xspace}

\makeatletter
\def\verbatim@font{\ttfamily\footnotesize}
\makeatother

\usepackage{xr-hyper}

\usepackage{booktabs}
\usepackage{longtable}
\usepackage{array}
\usepackage{multirow}
\usepackage{wrapfig}
\usepackage{float}
\usepackage{colortbl}
\usepackage{pdflscape}
\usepackage{tabu}
\usepackage{threeparttable}
\usepackage{threeparttablex}
\usepackage[normalem]{ulem}
\usepackage{makecell}
\usepackage{xcolor}
\usepackage[table]{xcolor}
\makeatletter
\@ifpackageloaded{caption}{}{\usepackage{caption}}
\AtBeginDocument{%
\ifdefined\contentsname
  \renewcommand*\contentsname{Table of contents}
\else
  \newcommand\contentsname{Table of contents}
\fi
\ifdefined\listfigurename
  \renewcommand*\listfigurename{List of Figures}
\else
  \newcommand\listfigurename{List of Figures}
\fi
\ifdefined\listtablename
  \renewcommand*\listtablename{List of Tables}
\else
  \newcommand\listtablename{List of Tables}
\fi
\ifdefined\figurename
  \renewcommand*\figurename{Figure}
\else
  \newcommand\figurename{Figure}
\fi
\ifdefined\tablename
  \renewcommand*\tablename{Table}
\else
  \newcommand\tablename{Table}
\fi
}
\@ifpackageloaded{float}{}{\usepackage{float}}
\floatstyle{ruled}
\@ifundefined{c@chapter}{\newfloat{codelisting}{h}{lop}}{\newfloat{codelisting}{h}{lop}[chapter]}
\floatname{codelisting}{Listing}

\makeatother
\makeatletter
\@ifpackageloaded{caption}{}{\usepackage{caption}}
\@ifpackageloaded{subcaption}{}{\usepackage{subcaption}}
\makeatother
\usepackage{bookmark}
\IfFileExists{xurl.sty}{\usepackage{xurl}}{} 
\makeatletter
\@ifundefined{xmpquote}{}{}
\makeatother
\hypersetup{
  pdftitle={A Fully Specified Bayes Factor Range for Hypothesis Testing and Sensitivity Analysis in Process Tracing},
  pdfauthor={Matias López; Jake Bowers; Daniel Gajardo Cooper},
  colorlinks=true,
  linkcolor={NavyBlue},
  filecolor={Maroon},
  citecolor={NavyBlue},
  urlcolor={NavyBlue},
  pdfcreator={LaTeX via pandoc}}

\title{A Fully Specified Bayes Factor Range for Hypothesis Testing and
Sensitivity Analysis in Process Tracing}

\author{Matias López\thanks{Universidad Diego Portales, Escuela de
Ciencia Política. \texttt{matiaslopez.uy@gmail.com}} \and Jake
Bowers\thanks{University of Illinois Urbana-Champaign, Departments of
Political Science \&
Statistics. \texttt{jwbowers@illinois.edu}} \and Daniel Gajardo
Cooper\thanks{Pontificia Universidad Católica de Chile, Instituto de
Ciencia Política. \texttt{drgajardo@uc.cl}}}

\date{2026-11-06}

\usepackage{color}
\usepackage{fancyvrb}

\DefineVerbatimEnvironment{Highlighting}{Verbatim}{commandchars=\\\{\}}
\usepackage{framed}
\definecolor{shadecolor}{RGB}{241,243,245}
\newenvironment{Shaded}{\begin{snugshade}}{\end{snugshade}}

\newcommand{\AttributeTok}[1]{\textcolor[rgb]{0.40,0.45,0.13}{#1}}

\newcommand{\CommentTok}[1]{\textcolor[rgb]{0.37,0.37,0.37}{#1}}

\newcommand{\DecValTok}[1]{\textcolor[rgb]{0.68,0.00,0.00}{#1}}

\newcommand{\FunctionTok}[1]{\textcolor[rgb]{0.28,0.35,0.67}{#1}}

\newcommand{\NormalTok}[1]{\textcolor[rgb]{0.00,0.23,0.31}{#1}}

\newcommand{\OtherTok}[1]{\textcolor[rgb]{0.00,0.23,0.31}{#1}}

\newcommand{\SpecialCharTok}[1]{\textcolor[rgb]{0.37,0.37,0.37}{#1}}

\newcommand{\StringTok}[1]{\textcolor[rgb]{0.13,0.47,0.30}{#1}}

\NewDocumentCommand\citeproctext{}{}

\makeatletter
 \let\@cite@ofmt\@firstofone
 \def\@biblabel#1{}
 \def\@cite#1#2{{#1\if@tempswa , #2\fi}}
\makeatother
\newlength{\cslhangindent}
\newlength{\csllabelwidth}
\newenvironment{CSLReferences}[2] 
 {\begin{list}{}{%
  \setlength{\itemindent}{0pt}
  \setlength{\leftmargin}{0pt}
  \setlength{\parsep}{0pt}
  \ifodd #1
   \setlength{\leftmargin}{\cslhangindent}
   \setlength{\itemindent}{-1\cslhangindent}
  \fi
  \setlength{\itemsep}{#2\baselineskip}}}
 {\end{list}}
\usepackage{calc}

\theoremstyle{plain}
\newtheorem{proposition}{Proposition}
\theoremstyle{plain}
\newtheorem{theorem}{Theorem}
\theoremstyle{remark}
\AtBeginDocument{}

\begin{document}
\maketitle
\begin{abstract}
A Bayes factor can be used in process tracing to reason about the
probative value of evidence for a working hypothesis against a rival
\autocite{fairfield2022}. However, relying on the researcher's own sense
of probative strength may inflate confidence in \(H_1\)
\autocite{zaks2021updating}. We show how to fully specify a generative
model of observations to compute Bayes factors in process tracing:
First, researchers count the evidence. Second, they classify it as
favoring either hypothesis. Third, they formalize a data-generating
process that makes the evidence more expected given \(H_1\) and another
for a rival \(H_R\). The ratio of these probabilities yields a fully
specified Bayes factor. We compute two Bayes factors under different
assumptions, yielding a Bayes factor range that allows researchers to
decide whether the balance of evidence favors \(H_1\) considering
probability thresholds. Our method also facilitates sensitivity analyses
of common assumptions, including those about bias and the probative
weight of a smoking gun.
\end{abstract}

\section{Introduction}\label{introduction}

Process tracing is a method for testing hypotheses in case studies based
on how well qualitative evidence fits different arguments
\autocite{collier2011,falleti2006a,george2005a,gerring2004,hall2013,mahoney2012logic}.
A Bayes factor is the ratio of the probability of the observed data
under one hypothesis to its probability under a competing hypothesis
\autocite{jeffreys1961theory,kass1995bayes}. Scholars, notably Fairfield
and Charman \autocite*{fairfield2022}, have argued for using Bayes
factors in process tracing to reason about the probative value of
evidence given different hypotheses \autocites[see
also][]{barrenechea2019,bennett2008process,rohlfing_bayes_2026,humphreys2023integrated}.
While the rationale for using Bayes factors is intuitive in qualitative
research, the main point of contention is how to represent the probative
value of evidence via the probability of qualitative observations.
Fairfield and Charman \autocite*{fairfield2022} note that we can express
the probative value of evidence with decibels, which can then be used to
compute a Bayes factor. Humphreys and Jacobs
\autocite*{humphreys2023integrated} focus more directly on using Bayes
factors to integrate over uncertainty about the causal structure that
generated the evidence, also basing decisions on ideas about probative
strength. Meanwhile, Zaks \autocite*{zaks2021updating} argues that a
researcher's idea of probative strength can easily inflate a Bayes
factor and with it confidence in \(H_1\). A debate followed between
Bennett et al. \autocite*{bennett_reply} and Zaks
\autocite*{zaks2021updating,zaks2022return} about the pros and cons of
Bayesian reasoning in qualitative research, without addressing the root
problem: How can we specify a probability model that says how probable a
researcher's qualitative observations are under each hypothesis?

We offer a solution by relying not on subjective ideas about the
probative value of data, but on a less disputed quantity: the amount of
data in favor of each explanation. Based on counting observations, we
propose a generative model of observation, which produces the two
probabilities that inform a Bayes factor. By ``generative model'', we
mean a probability model that describes the observations that were
actually made as a subset of evidence that belongs to a larger pool of
all potential observations, which includes evidence that was not
observed. The model gives the observed count a different probability
depending on whether the theorized pool of evidence is meant to confirm
\(H_1\) or \(H_R\). We do not use the model to describe how qualitative
evidence is actually produced, but to give every pattern of observations
a researcher could have made a probability we can compute under
assumptions that inform us about the balance of evidence for a
hypothesis.\footnote{In this paper, we engage only with the simple case
  of two theories that motivate a given study: the ``working
  hypothesis'' \(H_1\) is the explanation the researchers believe, and
  the ``rival hypothesis'' \(H_R\) is an explanation that is not
  compatible with it. We see this work as a precursor to a multi-theory
  extension that builds on these foundations.} This produces a Bayes
factor that does not rely on our reasoning about the strength of any
given observation, but rather on probabilities of a collection of
observations that can be known under a set of simple assumptions. In
turn, this model gives researchers flexibility to ask and answer
questions about their own assumptions. Because the model states each
assumption, a researcher can change one assumption at a time, such as
how an observation was coded or how much more a smoking gun counts than
other evidence, and compute whether the conclusion changes.

Our model allows us to compute two Bayes factors: a uniform-weights
Bayes factor and a worst-case Bayes factor, forming a Bayes factor
range. To calculate the Bayes factor range, the researcher must do one
thing only: count how many observations were made in favor of \(H_1\)
and \(H_R\). The model then computes how likely it is that those
particular counts would have occurred in two contrasting scenarios. The
uniform-weights scenario treats all plausible outcomes as equally likely
and the worst-case scenario assumes the proportion most favorable to the
rival hypothesis. This range covers Bayes factors under different
assumptions and the researcher can claim \(H_1\) is correct based on two
probability thresholds: \(BF>1\) for the worst-case lower bound and
\(BF \geq 20\) for the uniform-weights scenario. If the range meets
those values, we know that the qualitative evidence remains more likely
given \(H_1\) under the most demanding assumptions and that the
probability of the same evidence is below 5\% given \(H_R\) in the
uniform-weights scenario.

With a fully specified model, we can also use sensitivity analyses to
calculate how much bias would be needed to change a conclusion. If the
Bayes factor range meets the thresholds and would remain so under
significant bias, one can more confidently decide in favor of \(H_1\).
Here, the researcher did not calibrate the summary with her ideas about
probative strength, because labeling observations as favoring \(H_1\) or
\(H_R\) is a judgment about direction, not strength.

As we will show, this does not eliminate ideas about the probative value
of specific observations entirely. If smoking gun evidence is found, a
researcher may attach a probative weight to this particular observation
and see how the Bayes factor responds. However, the models make it
perfectly transparent whether conclusions depend largely on calibration
or not.

Relying on counts makes it easier to think about probabilities in terms
of frequencies. Simply put, we take \(H_1\) to claim that more than half
of the pool of potential observations favors \(H_1\), and \(H_R\) to
claim that at least half favors \(H_R\). Reducing a case to two counts
simplifies what a researcher actually does during fieldwork, and the
counts are meant to complement case knowledge rather than replace it.
Most of the work of process tracing --- gathering evidence, reading it
closely, and judging what it means --- remains exactly what it was. But
the model adds a way to state a statistical summary of confidence in
that work.

\section{Why Bayes factors require models}\label{sec-BF}

A Bayes factor compares the likelihood of evidence \(E\) given a
probabilistic model of observation that represents a working theory
\(H_1\) versus one that represents a rival \(H_R\)
\autocite{kass1995bayes}, thus:

\[
\text{Bayes factor} \equiv \text{BF} := \frac{p(E \mid H_1)}{p(E \mid H_R)}.
\]

To illustrate, think of a coin. One hypothesis says it lands tails 75\%
of the time (\(H_1\): the coin is rigged), while a rival says it lands
tails 50\% of the time (\(H_R\): the coin is fair). Say the set of
observations, \(E\), consists of 10 flips, all of which land tails. A
binomial formula counts the number of ways zero heads can land among ten
flips, \(\binom{10}{0}\), and multiplies by the probability of any one
such sequence under the model. We use the two models to compute a Bayes
factor:

\[
\text{BF}
= \frac{p(E \mid H_1)}{p(E \mid H_R)}
= \frac{\binom{10}{0}(0.25)^{0}(0.75)^{10}}
       {\binom{10}{0}(0.50)^{0}(0.50)^{10}}
= \frac{0.0563}{0.0010}
\approx 57.7.
\]

A Bayes factor of 57.7 describes the data as nearly 58 times as probable
under \(H_1\) as under \(H_R\), suggesting \(H_1\) is correct and the
coin is rigged.

But coin flips and social phenomena behave very differently. So let us
apply Bayesian factors a la \textcite{fairfield2022} in the following
example. A researcher studying democratic breakdown in country
\countrylabel entertains two competing hypotheses:

\begin{itemize}
\tightlist
\item
  \(H_1\): Elites in country \countrylabel helped bring the
  authoritarian leader to power in order to prevent a left-wing
  government from forming. Institutions did not matter.
\item
  \(H_R\): The institutions of country \countrylabel failed to prevent
  the rise of the authoritarian leader. Elites did not matter.
\end{itemize}

Assume the collection \(E\) is formed by twelve observations from
interviews, newspaper articles, and archival documents interpreted by
the researcher to largely corroborate \(H_1\) over \(H_R\).

Fairfield and Charman \autocite*{fairfield2022} ask us to state each
observation's probative value directly using decibels. A decibel value
\(d\) is an easy to understand way to summarize the ratio
\(p(E \mid H_1)/p(E \mid H_R)\): Fairfield and Charman
\autocite*[130]{fairfield2022} convert this likelihood ratio to decibels
by taking its base-10 logarithm and multiplying by 10,
\(d = 10 \log_{10} p(E \mid H_1)/p(E \mid H_R)\), so a ratio of 1, which
favors neither theory, becomes 0 dB, and each tenfold increase in the
ratio adds 10 dB.\footnote{A smoking gun might be worth around 30 dB,
  while other positive evidence might be worth 3 to 7 dB. Fairfield and
  Charman give 30 dB to their own example of a smoking gun, the lame
  sheep in the Sherlock Holmes story ``Silver Blaze''
  \autocite*[528]{fairfield2022}. Because 30 dB is a ratio of
  \(10^{30/10} = 1{,}000\), 3 dB a ratio of about 2, and 7 dB a ratio of
  about 5, a 30 dB smoking gun is \(1{,}000/5 = 200\) to
  \(1{,}000/2 = 500\) times as probative as ordinary positive evidence
  worth 3 to 7 dB. After scoring each observation with a decibel value,
  Fairfield and Charman add up the decibel values of the observations
  \autocite*[128]{fairfield2022}. They then add one more term, the
  researcher's prior odds of \(H_1\) against \(H_R\) in decibels, which
  is \(10 \log_{10}\) of the odds the researcher held before seeing any
  evidence. The resulting total expresses how strongly the researcher
  favors \(H_1\) after seeing the evidence
  \autocite*[128, 135--136]{fairfield2022}. In our notation, if \(D\) is
  the sum of the observations' decibel values alone, the Bayes factor
  \(p(E \mid H_1)/p(E \mid H_R)\) is \(10^{D/10}\), which is the product
  of the observations' likelihood ratios.}

Suppose that 75\% of the evidence favors \(H_1\), while the remaining
25\% favors \(H_R\). With twelve total observations, these proportions
correspond to nine observations supporting \(H_1\) and three supporting
\(H_R\). Suppose further that one of the observations within the 75\%
favoring \(H_1\) is a ``smoking gun'': a secret memo entitled ``Opening
the way for the autocrat: a plan to prevent a left-wing
government.''\footnote{``Smoking gun evidence'' and ``smoking gun test''
  are terms widely used in process-tracing research for an observation
  that is highly conclusive even if not necessary to infer that a causal
  statement is correct \autocite{collier2011,vanevera1997guide}.} The
researcher chooses to represent the ratio
\(p(E \mid H_1)/p(E \mid H_R)\) by assigning a likelihood ratio of 5 to
1 to the ordinary pro-\(H_1\) evidence, 1 to 5 to the pro-\(H_R\)
evidence, and 50 to 1 to the smoking-gun memo. Thus, eight of the twelve
observations are ordinary evidence favoring \(H_1\), one is the smoking
gun favoring \(H_1\), and three favor \(H_R\). On Fairfield and
Charman's scale, the likelihood ratios of 5 to 1, 1 to 5, and 50 to 1
are \(10 \log_{10} 5 \approx 7\) dB, \(10 \log_{10} (1/5) \approx -7\)
dB, and \(10 \log_{10} 50 \approx 17\) dB.\footnote{Adding these decibel
  values instead of multiplying the ratios gives
  \(8 \times 7 + 17 - 3 \times 7 = 52\) dB, which is
  \(10^{52/10} \approx 158{,}000\). That differs from the 156,250 below
  only because each decibel value here is rounded to a whole number.}
When the observations are independent under each theory, the Bayes
factor is the product of the likelihood ratios of the individual
observations. Multiplying these likelihood ratios across the twelve
observations:

\[
BF =
\underbrace{5^{9-1}}_{\substack{75\% \text{ favor } H_1\\\text{minus one smoking gun}}}
\cdot
\underbrace{50}_{\text{smoking gun}}
\cdot
\underbrace{(1/5)^3}_{25\% \text{ favor } H_R}
= 156{,}250.
\]

The Bayes factor of 156,250 means that the evidence is 156,250 times
more likely under \(H_1\) than under \(H_R\). With flat priors, \(H_1\)
and \(H_R\) each start with probability 1/2, so the prior odds are 1 to
1 and the posterior odds equal the Bayes factor, 156,250 to 1. Odds of
156,250 to 1 are a posterior probability of
\(p(H_1 \mid E) = BF/(1+BF) = 156{,}250/156{,}251 \approx 0.9999936\).
This probability expresses near certainty about the truth of \(H_1\).

As Zaks \autocite*{zaks2021updating} would note, the researcher's
calibration produces this result.\footnote{The eight ordinary
  pro-\(H_1\) observations at 5 to 1, after the three pro-\(H_R\)
  observations at 1 to 5 cancel three of them, contribute the factor
  \(5^5 = 3{,}125\), and the smoking-gun memo contributes the further
  factor of 50.} Say a peer objects to this calibration. The peer thinks
that most observations are 50\% more likely given \(H_1\) and that the
smoking gun is 100\% more likely. Now we have
\(BF = (1.5)^8 \cdot 2 \cdot (1/1.5)^3 = (1.5)^5 \cdot 2 \approx 15.2\).
With flat priors, the posterior probability of \(H_1\) falls from about
0.9999936 to about 0.94, and that of \(H_R\) rises from about 0.0000064
to about 0.06. The peer would still favor \(H_1\), at odds of about 15
to 1 rather than 156,250 to 1. Which BF should we use?

The researcher and the peer disagree about how strongly each observation
favors one hypothesis over the other. Nothing in the two hypotheses says
whether an ordinary observation is 5 or 1.5 times as probable under
\(H_1\) as under \(H_R\). But the researcher and the peer agree on
something: Both count nine observations favoring \(H_1\), the eight
ordinary ones and the memo, and three favoring \(H_R\). Their two Bayes
factors differ only because they give those twelve observations
different strengths. In the next section we build two Bayes factors from
the part they agree on. A Bayes factor is a ratio of the probabilities
of the observations under the two hypotheses. We cannot compute those
probabilities from the counts alone. In place of the strengths, we
compute those probabilities from what each hypothesis claims about the
pool of partially observed potential observations: under \(H_1\), more
than half of the pool favors \(H_1\), and under \(H_R\), at least half
favors \(H_R\).

\section{A binomial model of observation}\label{sec-binomial}

In our example, we count 12 observations and agree that 9 of those favor
\(H_1\) and 3 favor \(H_R\). Since 9 is more than 3, we could use these
counts to simply declare that \(H_1\) is more plausible than \(H_R\).
But the researcher could have interviewed other people or read other
documents. A different set of 12 observations could have had 3 favoring
\(H_1\) and 9 favoring \(H_R\) if the researcher happened to miss parts
of the archive or couldn't reach certain interview subjects. The two
hypotheses make opposite claims: \(H_1\) implies that data collection
will keep favoring \(H_1\), which tells us that there must be more data
favorable to \(H_1\) in the pool. Meanwhile, \(H_R\) claims the
opposite: digging further would reveal more evidence supporting \(H_R\)
instead. That is, the hypotheses themselves have implications for
observations not yet made.

Let \(\theta\) be the proportion of the pool that favors \(H_1\),
counting every observation the researcher \emph{could have made}, not
only the 12 the researcher actually collected. We do not know
\(\theta\), because the researcher has seen only those 12. If \(H_1\) is
the best explanation, then most of the available evidence should support
it. That is, under \(H_1\), more than half of the possible observations
favor \(H_1\), so \(\theta > 1/2\). We begin building the model with
\(\theta\) and turn this implication into a provisional assumption that
we return to when we discuss sensitivity analysis for observation bias
in Section~\ref{sec-sensitivity}.

If the researcher could continue interviewing indefinitely, or an
evolving political process keeps producing new documents, then
\(\theta\) is the proportion of pro-\(H_1\) evidence in a pool of
possible observations that keeps growing. But qualitative researchers
often work with finite evidence, such as a closed archive, and then
\(\theta\) is the proportion in that closed collection. The two forms of
evidence differ in how one document removed for observation changes the
proportion of pro-\(H_1\) evidence among the remaining documents. In a
pool that keeps growing, one document removed for observation barely
changes the proportion of pro-\(H_1\) evidence among the remaining
documents. In a closed archive of 10 documents with 6 pro-\(H_1\), the
proportion of pro-\(H_1\) evidence starts at 6/10. After the researcher
removes one pro-\(H_1\) document for observation, the proportion of
pro-\(H_1\) evidence among the 9 remaining documents falls to 5/9. We
use the same binomial model for both forms of evidence, and the Online
Supplement explains why that model is enough.

Our simplest working model assumes that each observation independently
supports \(H_1\) with probability \(\theta\).\footnote{The independence
  assumption is a working idealization. Fairfield and Charman make a
  similar concession, noting that ``proceeding as though the evidence is
  independent will be a reasonable approximation'' in many qualitative
  settings \autocite[116]{fairfield2022}.} This means that the evidence
supports \(H_R\) with probability \(1 - \theta\). Given a value of
\(\theta\), the model now can say how probable each total count is. Let
\(N\) be the total number of observations counted, \(k\) the number
supporting \(H_1\), and \(r = N - k\) the number supporting \(H_R\).
Given \(\theta\), the count \(k\) follows a binomial distribution,

\begin{equation}\protect\phantomsection\label{eq-binom}{p(E \mid \theta)=\binom{N}{k} \theta^{k}(1-\theta)^{r}.}\end{equation}

The researcher found that 9 of the 12 observations favor \(H_1\) (elite
choice) and 3 favor \(H_R\) (weak institutions). For example, if
\(\theta\) had single values of \(0.3\) (more evidence favoring \(H_R\))
or \(0.7\) (more evidence favoring \(H_1\)), the probability of
observing this exact split would be:

\[p(E \mid \theta = 0.3)=\binom{12}{9}0.3^{9}0.7^3\approx 0.001.\]
\[p(E \mid \theta = 0.7)=\binom{12}{9}0.7^{9}0.3^3\approx 0.24.\]

But neither theory names a single value of \(\theta\). \(H_1\) says
\(\theta\) is above one half, and the rival says it is one half or less,
and the probability of nine of twelve differs from one value of
\(\theta\) to another, as the two examples show. To compute a Bayes
factor, we need one probability for each theory, so we have to say how
to get from a range of values of \(\theta\) to a single number.

\subsection{Specifying models: the uniform-weights and worst-case Bayes
factors}\label{sec-two-bfs}

The Bayes factor computes the probability of the counts under \(H_1\)
(the numerator) divided by their probability under \(H_R\) (the
denominator). \(H_1\) implies that \(\theta\) is above one half, and
\(H_R\) implies that it is one half or less. For each hypothesis we have
to decide how to get one probability of the assembled observations from
a range of values of \(\theta\). The models we propose represent the
outcome of an agreement between the researcher, a proponent of \(H_1\),
and a critic, a proponent of \(H_R\).

The researcher would like to use \(\theta = 9/12\) for \(H_1\), the
value that makes nine of twelve most probable.\footnote{The probability
  of nine of twelve is \(\binom{12}{9}\theta^{9}(1-\theta)^{3}\).
  Raising \(\theta\) makes each of the nine observations that favor
  \(H_1\) more probable and each of the three that favor \(H_R\) less
  probable. Suppose \(\theta\) rises by a small amount \(d\). Each of
  the nine factors \(\theta\) grows by the fraction \(d/\theta\), so
  their product \(\theta^{9}\) grows by about nine times that fraction,
  to about \((1 + 9d/\theta)\) times its old value. (For example,
  \(1.01^{9} = 1.094\), close to \(1 + 9 \times 0.01 = 1.09\).) In the
  same way, \((1-\theta)^{3}\) shrinks to about \((1 - 3d/(1-\theta))\)
  times its old value. The probability is the product of these two
  parts, so it rises when the growth of the first, \(9d/\theta\), is
  larger than the shrinkage of the second, \(3d/(1-\theta)\). That is
  the case as long as \(9/\theta\) is larger than \(3/(1-\theta)\), that
  is, as long as \(9(1-\theta)\) is larger than \(3\theta\), or
  \(\theta\) is below \(9/12\). Above \(9/12\) the probability falls. It
  is 0.24 at \(\theta = 0.7\), 0.26 at \(\theta = 9/12\), and 0.24 at
  \(\theta = 0.8\). With \(k\) and \(r\) in place of 9 and 3, the same
  argument shows that the probability of the counts is largest at
  \(\theta = k/N\), the share of observations that favor \(H_1\).} But
the critic objects that the researcher chose, after seeing the counts,
the value most favorable to \(H_1\). A critic might like to use a value
near 1 to make the Bayes factor close to zero. In the face of these
conflicting incentives, the researcher commits, before seeing any
evidence, to averaging the probability of the counts over every value of
\(\theta\) above one half. Both parties find this a reasonable choice.
Both of our Bayes factors use this average as their numerator.

They differ in how they treat the rival. For the worst-case Bayes
factor, the researcher gives the critic the benefit of the doubt. The
best value of \(\theta\) for the model of \(H_R\) in the denominator of
the Bayes factor that is consistent with observing nine working theory
supporting observations is \(\theta = 1/2\): among the values at or
below one half, \(\theta = 1/2\) makes nine of twelve most probable. It
gives the largest denominator and so the smallest Bayes factor. For the
uniform-weights Bayes factor, the researcher counters by proposing that
the rival be treated the same way as \(H_1\): by averaging over every
value of \(\theta\) at or below one half. When the counts favor \(H_1\),
as ours do, this denominator averages probabilities that are each no
larger than the probability at \(\theta = 1/2\). So the uniform-weights
denominator is no larger than the worst-case denominator, and the
uniform-weights Bayes factor is never smaller than the worst-case Bayes
factor.\footnote{When the count favoring \(H_1\) is at least the count
  favoring \(H_R\), putting the rival at \(\theta = 1/2\) gives a
  smaller Bayes factor than any weights spread over the values at or
  below one half; the Online Supplement proves this (see ``The
  worst-case Bayes factor bounds the Bayes factor from below''). Taking
  the worst case on the working theory's range as well would give a
  Bayes factor near zero whenever even one observation favors \(H_R\),
  because weights concentrated near \(\theta = 1\) make such counts
  almost impossible.} That is, we have now two Bayes factors that
represent two agreements between proponents of the two theories under
scrutiny.

Averaging across the values of \(\theta\) raises a new question: how
much weight to give each value of \(\theta\) in the average. By default,
we propose that the researcher give every value of \(\theta\) between
one half and one the same weight in order to satisfy the critic without
unduly underweighting their observations.\footnote{Other weights are
  possible. In the Online Supplement, ``Choosing the distributions of
  \(\theta\) within each range'' explains what weights piled near one
  half or near 1 mean, and ``Sensitivity to the weights on \(\theta\)''
  reports the uniform-weights Bayes factor under other weights.}
Formally, equal weights give a uniform distribution on
\(\theta \in (1/2,1)\), with density \(p(\theta\mid H_1)=2\). These
weights are a prior distribution for \(\theta\) under \(H_1\): they say
how much each value of \(\theta\) above one half counts in the average.
They are not a prior probability of \(H_1\) itself. A researcher who,
before seeing any evidence, thought \(H_1\) (elite choice) twice as
probable as \(H_R\) (weak institutions) would hold prior odds of 2 to 1
for \(H_1\) over \(H_R\). The Bayes factor that we calculate here, so
far, does not use those odds. Combined with even prior odds, equal
weights on both sides make every count from 0 to \(N\) equally probable
before the researcher collects any evidence. With twelve observations,
each of the thirteen possible counts has probability 1/13.
\textcite{bayes1763} used this property to justify equal weights
\autocite{stigler1982}.

With these uniform weights, the numerator of both Bayes factors is
\begin{equation}\protect\phantomsection\label{eq-num}{p(E\mid H_1)=\displaystyle\int_{0.5}^1 \binom{N}{k} \theta^{k}(1-\theta)^{r} p(\theta\mid H_1)\,d\theta=2\displaystyle\int_{0.5}^1 \binom{N}{k} \theta^{k}(1-\theta)^{r}\,d\theta.}\end{equation}

For the uniform-weights Bayes factor, we treat \(H_R\)'s claim the same
way we did \(H_1\): every value of \(\theta\) from zero to one half gets
the same weight, so \(p(E \mid H_R)\) is the average of the binomial
probability over those values,
\begin{equation}\protect\phantomsection\label{eq-den-unif}{p(E\mid H_R)=\displaystyle\int_{0}^{0.5} \binom{N}{k} \theta^{k}(1-\theta)^{r} p(\theta\mid H_R)\,d\theta=2\displaystyle\int_{0}^{0.5} \binom{N}{k} \theta^{k}(1-\theta)^{r}\,d\theta, }\end{equation}
and the Bayes factor is the ratio of those two probabilities.

For the worst-case Bayes factor, the critic puts the rival at
\(\theta = 1/2\), so the denominator is one binomial probability,
\begin{equation}\protect\phantomsection\label{eq-den-worst}{p(E\mid H_R)=\binom{N}{k} (0.5)^{N}.}\end{equation}

Table~\ref{tbl-bf} shows the results of applying these models to our
example. Granting \(H_1\) for the sake of argument, we should find the
observed counts with a probability of 0.147. For the uniform-weights
Bayes factor, the rival's values of \(\theta\) get equal weights, as
\(H_1\)'s do, and the probability of the counts under \(H_R\) is
0.00710. The counts are then about 20.67 times as probable under \(H_1\)
as under \(H_R\). For the worst-case Bayes factor, the critic puts the
rival at \(\theta = 1/2\), where the probability of the counts is
0.0537, and the counts are only 2.73 times as probable under \(H_1\) as
under \(H_R\).

\begin{table}[H]

\caption{\label{tbl-bf}Bayes factors when 9 out of 12 observations
support \(H_1\)}

\centering{

\centering
\begin{threeparttable}
\begin{tabular}{lcccc}
\toprule
Bayes factor & $(k, r)$ & $P(E\mid H_1)$ & $P(E\mid H_R)$ & $BF=\frac{P(E\mid H_1)}{P(E\mid H_R)}$\\
\midrule
Uniform-weights & (9,3) & 0.147 & 0.00710 & 20.67\\
Worst-case & (9,3) & 0.147 & 0.0537 & 2.73\\
\bottomrule
\end{tabular}
\begin{tablenotes}[para]
\item Bayes factors are computed before rounding. The closed forms below give them exactly.
\end{tablenotes}
\end{threeparttable}

}

\end{table}%

Although \(p(E\mid H_1)\) is an average taken over a continuum of values
of \(\theta\), both Bayes factors can be computed directly from whole
numbers, without evaluating any integral. Write
\[S \;=\; \sum_{j=0}^{k}\binom{N+1}{j}.\] The numerator both Bayes
factors share has the closed form \(p(E\mid H_1) = S/(2^N(N+1))\)
(derived in the Online Supplement). Combining this with each denominator
gives closed forms for both Bayes factors:
\begin{equation}\protect\phantomsection\label{eq-closed}{\text{uniform-weights BF} = \frac{S}{2^{N+1}-S}, 
\qquad
\text{worst-case BF} = \frac{S}{(N+1)\binom{N}{k}}.}\end{equation}

At nine of twelve,
\(S = \sum_{j=0}^{9}\binom{13}{j} = 1+13+78+286+715+1287+1716+1716+1287+715 = 7814\).
Since \(2^{13}=8192\), the remaining four coefficients of row 13 sum to
\(8192-7814=378\), and \((N+1)\binom{N}{k} = 13\times220=2860\), giving
\[\text{uniform-weights BF} = \frac{7814}{378}\approx20.67,
\qquad
\text{worst-case BF} = \frac{7814}{2860}\approx2.73,\] matching the
values found above. A researcher can therefore check either Bayes factor
at any count by hand, from a table of binomial coefficients, without
software.

Figure~\ref{fig-bounded} shows the three probabilities for every count
the researcher could have found, from no observations favoring \(H_1\)
to all twelve. The probability under \(H_R\) averaged over \(\theta\) at
or below one half is highest at low counts and falls toward zero at high
counts, because the average includes values of \(\theta\) near zero,
where nine of twelve is almost impossible. The probability under \(H_1\)
reverses this pattern: lowest at a count of zero and highest at twelve.
The probability under \(H_R\) at \(\theta = 1/2\) is highest at six of
twelve. At the observed nine, the probability at \(\theta = 1/2\) is
0.0537 and the averaged probability is 0.00710, so the worst-case
denominator is about 0.0537/0.00710 \(\approx\) 7.6 times the
uniform-weights denominator. The uniform-weights Bayes factor, 20.67, is
therefore about 7.6 times the worst-case Bayes factor, 2.73.

\begin{figure}[!h]

\centering{

\pandocbounded{\includegraphics[keepaspectratio]{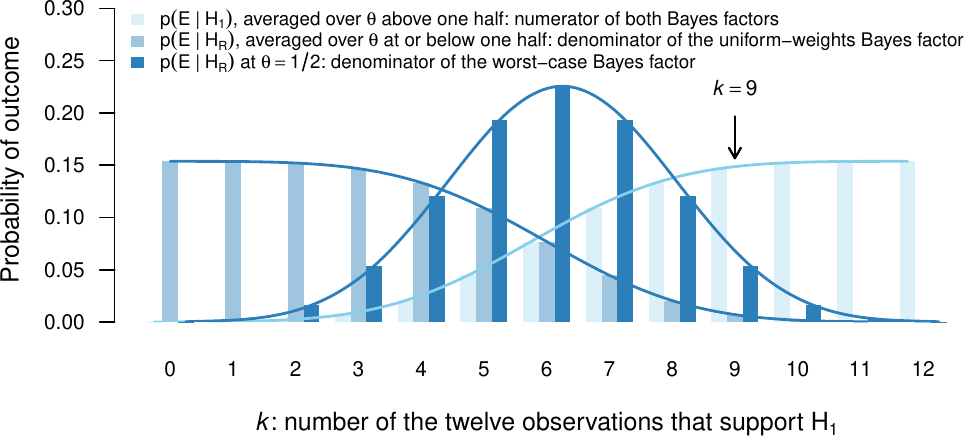}}

}

\caption{\label{fig-bounded}\textbf{Probability of each possible count
under \(H_1\) and under the two \(H_R\) models.} \emph{Note:} Each group
of bars represents one possible outcome, from none of the twelve
observations supporting the working theory to all twelve. Light bars
average over values of \(\theta\) above one half and form the numerator
shared by both Bayes factors. Medium bars average over values of
\(\theta\) at or below one half and form the denominator of the
uniform-weights Bayes factor. Dark bars evaluate the rival at
\(\theta = 1/2\), its best case and the denominator of the worst-case
Bayes factor. Light-blue curves trace the corresponding distributions
across possible values of \(k\).}

\end{figure}%

The uniform-weights Bayes factor is a Bayes factor as defined by
\textcite{kass1995bayes}: for each hypothesis, the probability of the
counts is averaged over its values of \(\theta\) given the chosen
weights. The resulting Bayes factor divides one average by the other.
The worst-case Bayes factor keeps the same averaged-over-\(\theta\)
numerator and uses the largest denominator that any weights on the
rival's values of \(\theta\) can give. When at least as many
observations favor \(H_1\) as favor \(H_R\), that largest denominator is
the probability at \(\theta = 1/2\) (see the proof in the Online
Supplement). So the worst-case Bayes factor is a lower bound on the
Bayes factors we could obtain by changing only the rival's weights. In
our example, the uniform-weights Bayes factor is about 7.6 times the
worst-case Bayes factor. Because the two can differ this much, we
propose in the next section a rule under which the researcher uses both
of them to decide whether to set the rival aside.

\subsection{Using the two Bayes factors to make a
decision}\label{sec-decide}

A Bayes factor greater than one tells us that the evidence is more
expected given \(H_1\), and a Bayes factor of 20 or more is a
conventional threshold for sufficiently ``strong'' evidence in the
statistics literature \autocite{kass1995bayes}\footnote{With equal prior
  odds, when the Bayes factor is 20 or more, the posterior probability
  of the rival, \(p(H_R \mid E) = 1/(1 + \text{BF})\), is at most
  \(1/21 \approx 0.048\)}. Let us build on these two thresholds.

The researcher in our example now has two Bayes factors: the
uniform-weights Bayes factor, 20.67, and the worst-case Bayes factor,
2.73, forming a Bayes factor range. These are informative in different
ways. We take the threshold proposed by \textcite{kass1995bayes} for the
uniform-weights Bayes factor point in the range. With flat priors, a
researcher who computed \(BF \geq 20\) finds that the posterior
probability of \(H_R\) is lower than 5\%.

On the lower bound of the Bayes factor range we have the worst-case
Bayes factor. Expecting a worst-case scenario to provide such a small
posterior for \(H_R\) may be too demanding, but a true \(H_1\) should
make the evidence at least more likely even in the most demanding
scenario. So we require the worst-case Bayes factor to exceed \(BF=1\)
before we can assume \(H_1\) as the correct answer. If the worst-case
Bayes factor exceeds 1, a critic cannot propose weights for the rival
that make the counts more probable under \(H_R\) than under \(H_1\).

\begin{figure}[H]

\centering{

\pandocbounded{\includegraphics[keepaspectratio]{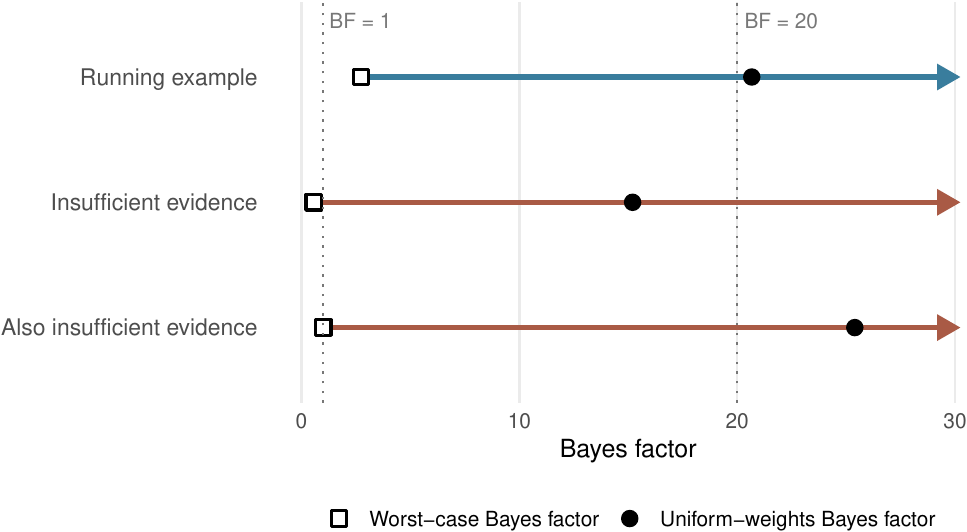}}

}

\caption{\label{fig-interval}Each line shows the Bayes factors that
different weights on the rival's values of \(\theta\) would give, with
the weights on \(H_1\)'s values held equal. The open square marks the
worst-case Bayes factor. It is the smallest Bayes factor that any
weights on the rival's values give, so each line starts at the square.
Near \(\theta = 0\) the observed counts are almost impossible. Weights
that favor the rival's values near zero therefore make the Bayes factor
as large as one likes. Each line ends in an arrow to show that it
continues past the edge of the plot. The filled circle marks the
uniform-weights Bayes factor. Equal weights on the rival's values give
this Bayes factor. The row labeled `Running example' has nine of twelve
observations favoring \(H_1\). The other two rows are invented to show
the two ways a researcher who applies both criteria would keep the
rival. In the row labeled `Insufficient evidence', the uniform-weights
Bayes factor is below 20. In the row labeled `Also insufficient
evidence', the uniform-weights Bayes factor is above 20 but the
worst-case Bayes factor is 1 or less.}

\end{figure}%

Figure~\ref{fig-interval} illustrates the two decision criteria. In the
running example, the uniform-weights Bayes factor exceeds 20 and the
worst-case Bayes factor exceeds 1. The other two rows illustrate why we
require both: the uniform-weights Bayes factor may fall below 20, or it
may exceed 20 while the worst-case Bayes factor is only 1.

Under the first criterion, the uniform-weights Bayes factor must be 20
or more. That Bayes factor depends on our choice of equal weights for
the values of \(\theta\). Would a researcher who chose other weights
reach a different decision? Under \(H_1\), equal weights put half the
weight on values of \(\theta\) between one half and three quarters and
half on values above three quarters, because the two ranges are equally
wide. But a researcher might think it likely that, even if elites
brought the leader to power, at least a quarter of the evidence would
point to weak institutions. One way to say so is to put about 71 percent
of the weight under \(H_1\), rather than half, on values of \(\theta\)
between one half and three quarters. The same researcher might think it
likely that, even if the country's institutions failed to prevent the
leader's rise, at least a quarter of the evidence would point to elites.
That researcher might then put about 71 percent of the weight under
\(H_R\) on values of \(\theta\) between one quarter and one half. With
both of these weights, the Bayes factor at nine of twelve is 7.35
instead of 20.67. That researcher keeps the rival. Another researcher
might think it likely that more than three quarters of the evidence
would point to elites if elites brought the leader to power, and to weak
institutions if the institutions failed. With about 71 percent of the
weight on values of \(\theta\) above three quarters under \(H_1\) and
below one quarter under \(H_R\), that researcher gets
BF=28.08.\footnote{The first of these researchers rescales a Beta(1/2,
  1) distribution onto the values of \(\theta\) from one half to one for
  \(H_1\), and a Beta(1, 1/2) onto the values from zero to one half for
  \(H_R\). Under the first distribution, the weight on values from one
  half to three quarters is \(\sqrt{1/2} \approx 0.71\). Under the
  second, the weight on values from one quarter to one half is also
  about 0.71. The second researcher uses Beta(1, 1/2) for \(H_1\) and
  Beta(1/2, 1) for \(H_R\). The Online Supplement explains rescaled Beta
  distributions and reports the Bayes factor under five such weightings.}

\subsection{Probative weight}\label{sec-weight}

The Bayes factors we have computed inform us about probabilities of the
observations made during the research project as a whole and not about
the probative weight of evidence of any one observation. In this
section, we begin to elaborate on the basic model. For example, we show
here how a researcher who interprets one piece of information as
extremely strong evidence can use the fully specified Bayes factor
approach to represent that judgment.

In our example, the researcher found a secret memo that describes a
``smoking gun'' for \(H_1\). The memo should carry more weight than
other less decisive observations in building evidence in favor of
\(H_1\) and against \(H_R\). When the researcher assigns weight
\(w_i \geq 1\) to observation \(i\), we compute both Bayes factors as if
multiple copies of that observation had been observed. For example, if
the assigned weight is \(10\), we count the observation as \(10\)
routine observations.\footnote{We restrict weights to positive integers
  so that every weighted count is a whole number and both Bayes factors
  are computed exactly. A researcher who wants a non-integer weight, say
  two and a half routine observations, rounds it to the nearest whole
  number. Multiplying every weight by a common factor to make them whole
  would not do: it changes the counts, and with them both Bayes factors.
  We think of ``routine'' observations as observations which we'd like
  to count as having more or less equal probative weight.} The weighted
counts, \[
W = \sum_{i \in H_1} w_i, \qquad R = \sum_{j \in H_R} w_j,
\] replace the counts \(k\) and \(r\) in Equation~\ref{eq-binom}. Their
sum, \(W + R\), replaces the number of observations \(N\). With unit
weights (\(w_i=1\)), \(W = k\) and \(R = r\), so we recover the
uniform-weights Bayes factor of 20.67 and the worst-case Bayes factor of
2.73. Larger weights allow unusually probative observations to count as
several routine observations, while leaving the model of observation
unchanged. We then compute the two Bayes factors from \(W\) and \(R\)
with Equations \ref{eq-num}, \ref{eq-den-unif}, and \ref{eq-den-worst},
or by hand with the closed forms in Equation~\ref{eq-closed}. When \(W\)
is at least \(R\), the worst-case Bayes factor computed from \(W\) and
\(R\) is still the smallest Bayes factor that any weights on the rival's
values of \(\theta\) give.

A disagreement about a weight is a disagreement about how many routine
observations the memo counts as, not about how many times as probable
the memo is under \(H_1\) as under \(H_R\). Because a weight changes
only the counts, each Bayes factor computed from \(W\) and \(R\) is
still the probability of the counts under the model for \(H_1\) divided
by their probability under a model for \(H_R\). A researcher who doubts
a weight of 10 can choose another weight and compute both Bayes factors
again.

For our country \countrylabel example, the researcher identifies the
secret memo as a smoking gun and assigns it a weight of 10. The
remaining eight pro-\(H_1\) observations and the three pro-\(H_R\)
observations keep their unit weights. The weighted counts are
\(W = 8 + 10 = 18\) and \(R = 3\).

\begin{table}[H]

\caption{\label{tbl-weighted-bf}Bayes factors before and after
weighting.}

\centering{

\centering
\begin{tabular}{lccc}
\toprule
Bayes factor & $(k, r)$ & Unweighted & Memo at weight 10\\
\midrule
Uniform weights & (9, 3) & 20.67 & 2,337\\
Worst case & (9, 3) & 2.73 & 143.3\\
\bottomrule
\multicolumn{4}{l}{\rule{0pt}{1em}\textit{Note:} One pro-$H_1$ observation (a memo) assigned weight $w = 10$.}\\
\end{tabular}

}

\end{table}%

Table~\ref{tbl-weighted-bf} shows how the two Bayes factors change when
the researcher gives the memo a weight of 10. The uniform-weights Bayes
factor rises from 20.67 to 2,337. The worst-case Bayes factor rises from
2.73 to 143.3. Even when the critic puts the rival at \(\theta = 1/2\),
as in Equation~\ref{eq-den-worst}, the weighted counts are 143.3 times
as probable under the model for \(H_1\) as under the model for \(H_R\).

\section{Assessing sensitivity to assumptions and
judgments}\label{sec-sensitivity}

Our two Bayes factors rest on five assumptions:

\begin{enumerate}
\def\labelenumi{\arabic{enumi}.}
\tightlist
\item
  \emph{Unbiased search}: Nothing in the research process makes
  pro-\(H_1\) observations more likely to be found than pro-\(H_R\)
  observations.
\item
  \emph{Independence}: Whether one observation favors \(H_1\) tells us
  nothing about whether the next one does.
\item
  \emph{Coding}: The researcher decides which hypothesis each
  observation favors.
\item
  \emph{Weights on observations}: The researcher decides the probative
  weight of each observation relative to the others.
\item
  \emph{Weights on \(\theta\)}: The researcher gives equal weight to
  every value of \(\theta\) above one half under \(H_1\). For the
  uniform-weights Bayes factor, the researcher also gives equal weight
  to every value at or below one half under \(H_R\). The worst-case
  Bayes factor needs no judgment about the rival's weights. When \(W\)
  is at least \(R\), \(\theta=1/2\) produces the smallest Bayes factor
  that any weights on the rival's values of \(\theta\) give.
\end{enumerate}

We cannot test those assumptions, but we can know how much violation of
each is needed to change our conclusion. In this section of the paper,
we focus mainly on the assumption researchers worry about the most:
whether the observations were collected in a way biased towards \(H_1\),
i.e.~input 1. But we also compare sensitivity to different weights on
observations, input 4. Both are common worries among researchers doing
process tracing. We do not return to the weights on \(\theta\), input 5,
here. In Section~\ref{sec-decide}, two other choices of weights on
\(\theta\) gave Bayes factors of 7.35 and 28.08 instead of 20.67. In the
Online Supplement, we compute the uniform-weights Bayes factor under
five choices of weights, among them the equal weights and the two above.

The assumption of unbiased search tells us that finding mostly
pro-\(H_1\) observations is a signal that if the researcher continued to
collect observations, they would find mostly other pro-\(H_1\)
observations. If bias was present, then something made the researcher
miss pro-\(H_R\) observations. Of course, missing a few pro-\(H_R\)
observations need not impugn an overall strong pro-\(H_1\) finding. The
sensitivity analysis that we develop shows how lopsided that search
would have had to be to change the substantive conclusions arising from
the Bayes factor.

We write \(\omega\) for the bias factor. A search with bias factor
\(\omega\) is \(\omega\) times as likely to find any single pro-\(H_1\)
observation as any single pro-\(H_R\) one. With \(\omega = 2\), for
example, a document pointing to the elites based explanation (\(H_1\))
is twice as likely to turn up in the researcher's search as a document
supporting the weak institutions explanation (\(H_R\)), perhaps because
the people who could describe the institutions' weakness were harder to
reach for an interview. If the share of the pool that supports \(H_1\)
is \(\theta\), the probability that any single observation the search
finds is pro-\(H_1\) is then
\(\omega\theta/(\omega\theta + 1 - \theta)\). If \(\omega = 1\), this
probability is \(\theta\) itself, so a search with no bias leaves the
model as it was. The Online Supplement gives the derivation. We put that
probability in place of \(\theta\) in Equation~\ref{eq-binom} and
compute both Bayes factors from Equations \ref{eq-num},
\ref{eq-den-unif}, and \ref{eq-den-worst} with that change. We write
\(\omega^\star\) for the smallest bias factor at which the researcher
would no longer set the rival aside, because the uniform-weights Bayes
factor has fallen to 20 or the worst-case Bayes factor has fallen to 1.

We find \(\omega^\star\) by increasing \(\omega\) from 1 and computing
both Bayes factors again until the uniform-weights Bayes factor falls to
20 or the worst-case Bayes factor falls to 1.\footnote{In ``Sensitivity
  to observation bias,'' the Online Supplement writes both Bayes factors
  as functions of \(\omega\) and solves for \(\omega^\star\)
  numerically.} Allowing for a larger bias lowers both Bayes factors,
because a search that finds pro-\(H_1\) observations more easily would
produce mostly pro-\(H_1\) counts even if most of the evidence that
could be found favored the rival. With each observation counted once,
the uniform-weights Bayes factor is 20.67, and it falls to 20 at
\(\omega^\star = 1.01\) (the dot at a weight of 1 in
Figure~\ref{fig-weight-bias}). A search only 1 percent more likely to
find a pro-\(H_1\) observation than a pro-\(H_R\) observation would
change the decision to set the rival aside. With the memo counted as ten
observations, \(\omega^\star\) is 2.33 (the dot at a weight of 10). The
decision would then change only if the search were more than twice as
likely to find a pro-\(H_1\) observation as a pro-\(H_R\) one. If the
search was not biased, the researcher sets the rival aside regardless of
the weight placed on the memo. This means that the researcher could
side-step debates about the importance of the memo: the memo need not be
a smoking-gun in order for the evidence in favor of \(H_1\) to be
strong.

\begin{figure}[H]

\centering{

\pandocbounded{\includegraphics[keepaspectratio]{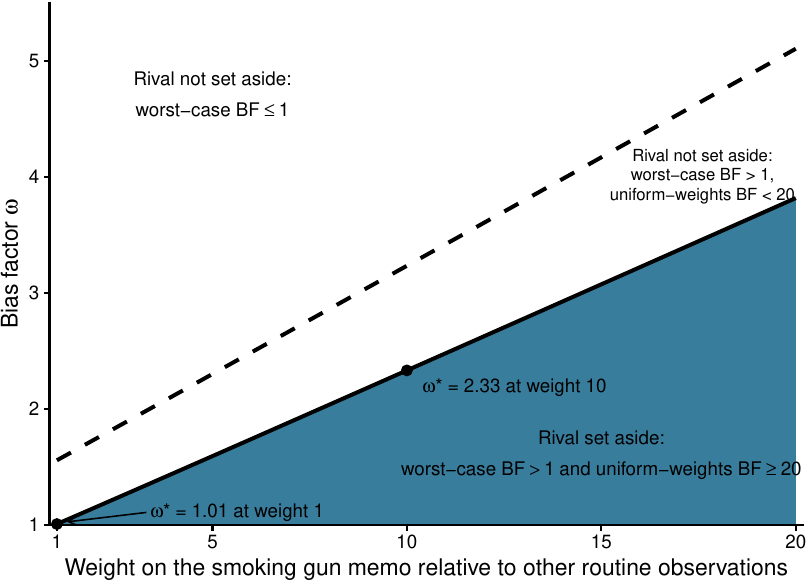}}

}

\caption{\label{fig-weight-bias}The relationship between search bias and
probative weight on the smoking-gun memo and evidence in favor of
\(H_1\). The solid line is the bias factor \(\omega\) at which the
uniform-weights Bayes factor equals 20. The dashed line is the bias
factor at which the worst-case Bayes factor equals 1. The two dots mark
the points where the memo has a weight of 1 versus 10 discussed in the
text.}

\end{figure}%

What the memo's weight determines is how much bias it takes to change
that decision. The solid line in Figure~\ref{fig-weight-bias} shows
\(\omega^\star\) at every weight from 1 to 20. Below the dashed line,
where the worst-case Bayes factor is above 1, the counts are more
probable under \(H_1\) than under the rival. With the memo counted as
ten observations, the dashed line is at 3.24. Between a bias of 2.33 and
3.24, the counts therefore still favor \(H_1\) under every choice of
rival weights, but not strongly enough for the researcher to set the
rival aside.

In the running example, the decision to set the rival aside holds only
if the search was close to unbiased, unless the memo counts for more
than one observation. If we did not treat the memo as a smoking gun, a
search only 1 percent more likely to find pro-\(H_1\) observations would
reverse the decision. With the memo counted as ten observations, the
decision holds until the search is more than twice as likely to find
pro-\(H_1\) observations as pro-\(H_R\) ones. In the Online Supplement,
we ask the same kind of question of the other inputs: how many
observations would have to be coded differently to change the decision,
how much treating correlated observations as independent overstates the
evidence, and how the uniform-weights Bayes factor changes under other
weights on \(\theta\). A sensitivity analysis does not tell us whether
the search was lopsided. The researcher who is savvy about the history
and context of the study can judge the plausibility of such bias from
what they know about how the documents were collected and who agreed to
be interviewed. The sensitivity analysis gives that researcher, and any
reader who doubts the result, a specific question to answer given their
expertise. The question here is not whether the search was biased at all
(obviously work in the archives goes from one document to the next in a
way that is purposive), but whether it could have been 2.33 times as
likely to find pro-\(H_1\) observations as pro-\(H_R\) ones.

\section{Applications in Published Work}\label{sec-applications}

To see how our approach performs, we had two independent AI agents code
the observations in six published papers (see the Online Supplement for
how we chose them). A \emph{charitable} coder took the author's framing
at face value and resolved ambiguity in favor of \(H_1\). A
\emph{skeptical} coder pushed back on the author's framing and resolved
ambiguity in favor of the rival.\footnote{Each agent backed every coding
  with a quote and a page citation. The coded observations, with their
  quotes and page citations, will be made available in the reproduction
  archive for this paper.} The two sets of coded material were then
merged under a consensus rule: an observation enters the evidence counts
(\(k\) for \(H_1\) and \(r\) for \(H_R\)) only when both coders
agree.\footnote{We set aside observations as ambiguous when the two
  coders disagreed about which theory the observation favors. The merge
  is conservative by design --- it loses signal from observations that
  one coder flags but the other does not --- but it forces the analyst
  to justify each consensus coding. Each observation was recorded in a
  row in a data table. The full per-row coding for all applications is
  at
  \texttt{replications/\textless{}paper\_key\textgreater{}/cases.csv}.}
We then checked the coding by hand and searched for evidence that most
researchers would probably classify as a smoking gun. The Online
Supplement's section for each study reports the coding and the
sensitivity analyses.

Here is a summary of the six studies. \textcite{steinsson2024} process
traces how Wikipedia became a trustworthy source of information, testing
a theory about struggles within the community of Wikipedia editors
versus a rival theory about causes outside Wikipedia.
\textcite{winward2021} tests a theory that credits the 1960s Indonesian
killings to low state capacity. The rival credits violence to the
strength of challenging groups. \textcite{gallego2023} put forward a
theory about the unexpected generosity of migration laws in Latin
America, claiming that symbolic politics pushed by liberal elites, not
immigrant stocks, accounted for it. \textcite{mor2022} attributes the
emergence of a Catholic party in 19th-century Prussia to voters
coordinating around Catholic identity under the threat of hostile
government policies, against the rival theory that sees the party as
pushed by elites. \textcite{andersen2024} tells us that the peaceful
agrarian reforms in the 18th and 19th centuries in Scandinavia were the
product of technocratic states rationalizing land distribution rather
than being driven by conflict between lords and peasants. None of these
five studies relies on a smoking-gun observation. Finally,
\textcite{pavone2022} argue that policy choices in Norway's 2019
social-benefits reform are explained by proponents trying to avoid
judicial review, and not by a recognition of ``legal obligation.'' Here,
the authors did find a smoking gun: a letter explicitly stating that the
threat of judicialization was the reason why policymakers changed their
policies. We report both Bayes factors for every study, and
Table~\ref{tbl-app-summary} collects the results.

\begin{table}[!h]

\caption{\label{tbl-app-summary}Six published process-tracing studies
under the uniform-weights Bayes factor and the worst-case Bayes factor.
Observation bias \(\omega^\star\) is the smallest factor by which the
researcher's search would have had to favor pro-\(H_1\) evidence for the
uniform-weights Bayes factor to fall below 20 or the worst-case Bayes
factor to fall to 1.}

\centering{

\centering\begingroup\fontsize{10}{12}\selectfont

\resizebox{\ifdim\width>\linewidth\linewidth\else\width\fi}{!}{
\begin{threeparttable}
\begin{tabular}{lrS[table-format=4.2]S[table-format=3.2]S[table-format=1.2]}
\toprule
Paper & $(k,r)$ & {Uniform weights} & {Worst case} & {$\omega^{\star}$}\\
\midrule
Steinsson 2024 & (12, 0) & 8191 & 630 & 5.49\\
Winward 2021 & (14, 3) & 264 & 21.34 & 1.74\\
Hammoud-Gallego \& Freier 2023 & (11, 2) & 154 & 14.91 & 1.70\\
Mor 2022 & (8, 2) & 29.57 & 4.00 & 1.14\\
Andersen 2024 & (9, 3) & 20.67 & 2.73 & 1.01\\
Pavone \& Stiansen 2022 & (7, 4) & 4.16 & 0.83 & {$-$}\\
\addlinespace
\hspace{1em}Pavone \& Stiansen (smoking gun) &  & 277 & 20.54 & 1.67\\
\bottomrule
\end{tabular}
\begin{tablenotes}
\item \textit{Note: } 
\item $(k,r)$ counts observations favoring $H_1$ and $H_R$. The uniform-weights Bayes factor averages the probability of the observed counts over values of $\theta$ consistent with each hypothesis using uniform weights. The worst-case Bayes factor evaluates the rival at $\theta=1/2$, the rival's best case, and therefore provides a lower bound on the Bayes factor against the rival. Observation bias $\omega^{\star}$ reports how much the observation process would have had to favor evidence supporting $H_1$ for the uniform-weights Bayes factor to fall below 20 or the worst-case BF to fall below 1 (whichever happens first). The Pavone \& Stiansen sub-row assigns weight 10 to the smoking-gun observation. Bayes factors below 100 are shown to two decimal places and larger ones to the nearest whole number.
\end{tablenotes}
\end{threeparttable}}
\endgroup{}

}

\end{table}%

\begin{figure}[H]

\centering{

\pandocbounded{\includegraphics[keepaspectratio]{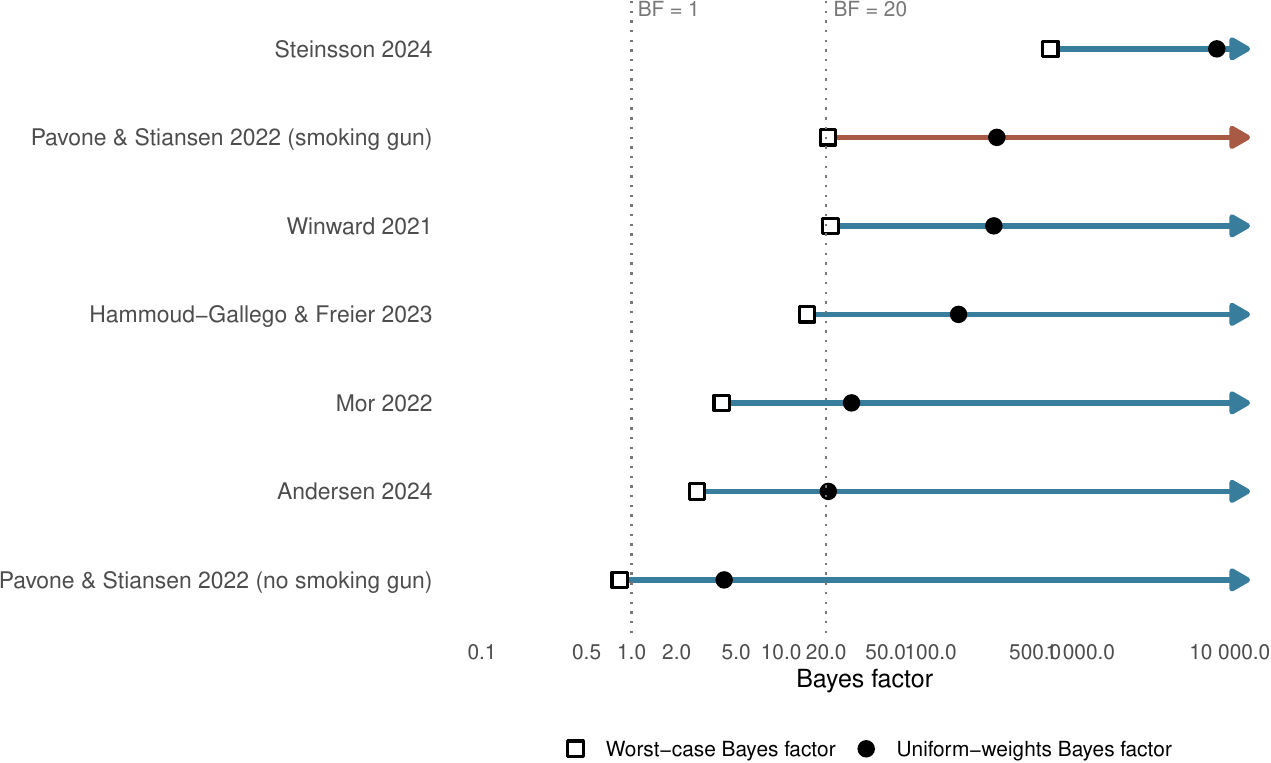}}

}

\caption{\label{fig-forest-bf}Bayes factors for six published
process-tracing studies. Open squares indicate the worst-case Bayes
factor and filled circles the uniform-weights Bayes factor. Arrows
indicate that alternative weights on the rival's values of \(\theta\)
can produce arbitrarily large Bayes factors. The smoking-gun
illustration assigns weight 10 to one observation.}

\end{figure}%

Figure~\ref{fig-forest-bf} illustrates the same information in relation
to the two thresholds. A researcher can immediately tell if the evidence
is sufficiently strong. Considering the weight of smoking gun evidence,
the lower bound worst-case scenario of all studies is greater than 1 and
the uniform weights central Bayes factor in all surpasses the value of
20.

Notably, a positive result depends on the smoking gun weight for Pavone
and Stiansen's study. However, we do not need the weight of 10 that was
assigned. At a weight of 5, the two Bayes factors are 25.04 and 2.89 and
the probability range meets the test. So the conclusion would survive a
reviewer who thinks that the probative value of that observation is half
of that assigned by the researcher.

The last column of Table~\ref{tbl-app-summary} tells us the bias needed
to flip the conclusion. Most studies admit significant bias before the
uniform weights Bayes factor drops below 20. The most robust conclusion
is that of Steinsson, which resists bias up to 449\%. The least robust
is that of Andersen, whose conclusions would change at 1\% bias.
Notably, we counted the same number of observations in both studies
(\(N=12\)), but all observations in Steinsson's study favor \(H_1\) and
some in Andersen's seem to fit \(H_R\). This is where the qualitative
research context matters as well. Steinsson can more easily make
observations about the Wikipedia process compared to Andersen's access
to information about a pre-modern political process. Their qualitative
data is quite different, but both could make a statistical summary using
a Bayes factor range and sensitivity analysis.

\section{Discussion and conclusion}\label{sec-discussion}

\textcite{bennett2008process} described process tracing in Bayesian
terms. Fairfield and Charman
\autocite*{fairfield2017explicit,fairfield2022} then developed the
explicit use of Bayes factors in process tracing, representing the
probative value of each piece of evidence by a likelihood ratio. Others
took up Bayesian process tracing with enthusiasm and developed
approaches of their own
\autocite{humphreys2015,barrenechea2019,humphreys2023integrated,rohlfing_bayes_2026}.
\textcite{zaks2021updating} criticized this literature sharply,
\textcite{bennett_reply} replied, and \textcite{zaks2022return} answered
that reply. We share both the enthusiasm and the concerns.

This paper adds three things researchers can use alongside the Bayes
factors they already compute. We compute Bayes factors from a fully
specified probability model of all the assembled observations, rather
than from probabilities assigned by hand to one observation at a time.
We offer a practical default rule for decision making based on two Bayes
factors from that model, which form a Bayes factor range that covers
different assumptions down to a worst-case scenario. And we offer
sensitivity analyses with which a researcher can ask how far each
assumption would have to fail, or each judgment differ, before the
decision to set the rival aside would change.

Our model is a model of agreement and disagreement between a proponent
of \(H_1\) and a proponent of the rival. An expert in the research
context interprets each observation as counting for one side or the
other. Before sensitivity analysis, our model requires one basic
agreement between the two sides: that \(k\) observations favor \(H_1\)
and \(r\) favor \(H_R\).\footnote{Our model currently excludes
  observations that favor neither side or both sides equally from the
  counts.} The two disagree about the pool of potential observations the
researcher draws from: the proponent of \(H_1\) claims that more than
half of it favors \(H_1\), and the proponent of the rival claims no more
than half does. The Bayes factor is the probability of the observed
counts under the first claim divided by their probability under the
second. With the sensitivity analyses of Section~\ref{sec-sensitivity},
a researcher can ask how much violation is needed to challenge the
conclusion.

A key difference in our model is that it relates to probabilities of
observation, not to strength. So the probative value and probability are
distinct quantities, and not the same as in Fairfield and Charman
\autocite*{fairfield2022}. But whereas we move more decisively towards
integrating statistical inference in process tracing, because we keep
these tasks separate we do not require researchers to change how they
scrutinize evidence qualitatively. Researchers can continue to rely on
interpretation to assess the strength of evidence. Computing a Bayes
factor to learn about the evidence represents an additional task that
provides different information that is also useful to make a final
decision about whether the balance of evidence favors a hypothesis.

But even fully specified Bayes factors carry problems. Statisticians
have pointed out a problem that applies to any Bayes factor computed
from a model in which a hypothesis allows more than one value of a
parameter. Such a Bayes factor depends on the weights the researcher
places on those values \autocite[182--184]{gelman2013bda3}.\footnote{\textcite[781--782]{kass1995bayes}
  write that Bayes factors ``tend to be sensitive to the choices of
  priors on the model parameters,'' and for that reason they recommend
  computing a Bayes factor ``over a range of possibilities.'' Kass and
  Raftery also describe the largest value a Bayes factor in favor of one
  hypothesis can take over all possible priors under that hypothesis.
  That largest value comes from putting all of the prior on the
  parameter value that makes the data most probable. Our worst-case
  Bayes factor is built the same way from the rival's side. When at
  least as many observations favor \(H_1\) as favor \(H_R\), the weights
  that make the observed counts most probable put all of the rival's
  weight at one half, and they give the smallest Bayes factor in favor
  of \(H_1\).} Our model is also sensitive to the weights, and we
designed the worst-case Bayes factor to confront this problem and
provided a framework for sensitivity analysis. If all criteria are met,
the statistical summary corroborates \(H_1\) more clearly.

Naturally, our model also has limitations. Process tracing often
proceeds from one observation to the next such that the second
observation depends on the nature of the first. Also, several documents
may all come from the same source. We explore the idea of dependent
observations in the Online Supplement, where we compute the
uniform-weights Bayes factor for our running example under one simple
model of dependence, in which every pair of observations is correlated
to the same degree. We find that even a small correlation changes the
Bayes factor. With every pair of observations correlated at 0.02, the
uniform-weights Bayes factor for nine of twelve falls from 20.67 to
13.78, and at a correlation of 0.30 it falls to 2.21. Another limitation
is that our model compares \(H_1\) with one rival at a time, and process
tracing often uses several rivals. Our coding also treats each
observation as favoring one theory or the other. Theories are often not
mutually exclusive, and an observation may bear on more than one of
them.

Questions about theories with overlapping implications, several rivals,
and a network of dependent observations deserve papers of their own.
Each needs an extension of the current model: a way to let one
observation bear on more than one theory, a model with more than one
rival, and a model that says which observations depend on which others
and in what way. This paper gives those papers their starting point: a
fully specified model and sensitivity analyses for the model's
assumptions and judgments. Our model cannot yet represent all the steps
in qualitative case study research, but it does illuminate a promising
path. Extending this approach would let more researchers state how
strongly their observations favor one theory and why, so that other
scholars can take up the argument, and build on it.

\section*{Acknowledgement of AI use}\label{acknowledgement-of-ai-use}
\addcontentsline{toc}{section}{Acknowledgement of AI use}

We used AI tools in multiple ways. First, we used Claude Code
(Anthropic) to search for studies using case studies and process
tracing, to winnow down that population to the set we study in the paper
following rules that we determined, and to code the observations in the
six published studies, as the two independent coders described in the
applications section. The team reviewed all of those codings. Second,
Claude (Anthropic) drafted, tested, and edited R code and drafted and
edited passages of the paper and supplement, including helping us stay
within the word counts required by journals. Third, Codex (OpenAI)
checked our math. Fourth, Claude Code helped us in many logistical ways:
we used it to fix LaTeX and Quarto errors, which enabled us to submit
the manuscript to arXiv, and to enhance our coordination via RStudio
Server running on an Amazon EC2 instance provisioned by the Prism system
for scientific computing, GitHub, and our Makefile. The authors reviewed
all AI output and are responsible for the content of this paper.

\printbibliography

\clearpage
\setcounter{secnumdepth}{5}
\setcounter{section}{0}
\setcounter{subsection}{0}
\setcounter{subsubsection}{0}
\ifcsname c@theorem\endcsname\setcounter{theorem}{0}\fi
\ifcsname c@proposition\endcsname\setcounter{proposition}{0}\fi
\ifcsname c@lemma\endcsname\setcounter{lemma}{0}\fi
\ifcsname c@corollary\endcsname\setcounter{corollary}{0}\fi
\phantomsection
\addcontentsline{toc}{section}{Online Supplement}
\begin{center}
{\LARGE\bfseries Online Supplement\par}
\end{center}
\vspace{1.5em}

\section{The Model and the Bayes
Factors}\label{the-model-and-the-bayes-factors}

This supplement provides formal details on the generative model
introduced in the main paper, the sensitivity analysis, and the
application of the method to six published studies.

This supplement follows the main paper in using \(k\) to denote the
count of pro-\(H_1\) observations and \(r\) to denote the count of
pro-\(H_R\) observations. We do depart from the main paper in one place:
where the main paper writes \(p(\cdot)\) for both probabilities and
densities, this supplement reserves \(f(\cdot)\) for densities of
continuous parameters (such as the posterior on \(\theta\) in the
binomial model) and \(p(\cdot)\) for probabilities and pmfs. We
apologize for the slight inconsistency; the supplement uses both kinds
of object in close succession in the binomial proofs, and a visible
\(f\)/\(p\) distinction makes the proofs easier to follow. The two files
are otherwise consistent in notation.

The model represents the evidence-generating process as a sequence of
independent draws, each supporting \(H_1\) with probability \(\theta\)
and \(H_R\) with probability \(1 - \theta\), where \(\theta\) is the
share of the evidence the researcher could have examined that supports
the working theory. We assume \(\theta>1/2\) under \(H_1\) and
\(\theta\leq 1/2\) under \(H_R\), and each theory carries a density over
its own range, \(f(\theta\mid H_1)\) on \((1/2, 1)\) and
\(f(\theta\mid H_R)\) on \((0, 1/2]\). If we collect \(N\) pieces of
evidence, \(k\) of them supporting \(H_1\), the uniform-weights Bayes
factor and the worst-case Bayes factor are
\[\text{uniform-weights BF}=\dfrac{\displaystyle\int_{0.5}^1 \theta^k(1-\theta)^{N-k}\, f(\theta\mid H_1)\,d\theta}{\displaystyle\int_{0}^{0.5} \theta^k(1-\theta)^{N-k}\, f(\theta\mid H_R)\,d\theta},\]
\[\text{worst-case BF}=\dfrac{\displaystyle\int_{0.5}^1 \theta^k(1-\theta)^{N-k}\, f(\theta\mid H_1)\,d\theta}{\left(0.5\right)^N},\]
the binomial coefficient having cancelled. Our default gives \(\theta\)
uniform weight within each range, and
Section~\ref{sec-choosing-distributions} says what other choices for the
uniform-weight Bayes factor mean.

\subsection{The closed forms}\label{sec-closed-forms}

When using uniform weights, both Bayes factors can be computed directly
from whole numbers; therefore, they can be computed manually without
evaluating any integrals.

\begin{theorem}[]\protect\hypertarget{thm-closed-form}{}\label{thm-closed-form}

Write \(S=\displaystyle \sum_{j=0}^{k}\binom{N+1}{j}\). Then
\[\text{uniform-weights BF} = \frac{S}{2^{N+1}-S}, 
\qquad
\text{worst-case BF} = \frac{S}{(N+1)\binom{N}{k}}.\]

\end{theorem}

\begin{proof}
Let \(F(x)=\displaystyle\int_0^x \theta^k(1-\theta)^{N-k}d\theta\). By
the relation between the incomplete beta function and the binomial
distribution (Abramowitz and Stegun 1964, sec. 26.5),
\[F(x)=\displaystyle\dfrac{\sum_{j=k+1}^{N+1}\binom{N+1}{j}x^j(1-x)^{N+1-j}}{(N+1)\binom{N}{k}}\]
Evaluating at \(1\) and \(\dfrac{1}{2}\), \[
F(1)=\dfrac{1}{(N+1)\binom{N}{k}}\qquad\text{and}\qquad F(\frac{1}{2})=\dfrac{1}{2^{N+1}(N+1)\binom{N}{k}}\sum_{j=k+1}^{N+1}\binom{N+1}{j}=F(1)\dfrac{2^{N+1}-S}{2^{N+1}}.
\] Hence,

\[
\text{uniform-weights BF} = \dfrac{F(1)-F(\frac{1}{2})}{F(\frac{1}{2})}=\dfrac{F(1)\dfrac{S}{2^{N+1}}}{F(1)\dfrac{2^{N+1}-S}{2^{N+1}}}=\frac{S}{2^{N+1}-S}
\] and \[
\text{worst-case BF} = \dfrac{2(F(1)-F(\frac{1}{2}))}{(\frac{1}{2})^N}=2^{N+1}F(1)\dfrac{S}{2^{N+1}}=\frac{S}{(N+1)\binom{N}{k}}
\]
\end{proof}

\subsection{Finite evidence without
replacement}\label{finite-evidence-without-replacement}

The binomial formula assumes that the chance that the next observation
supports \(H_1\) stays at \(\theta\), whatever the researcher has
already observed. In a finite archive that chance changes. If 6 of 10
documents support \(H_1\) and the first document drawn at random
supports \(H_1\), then 5 of the 9 documents left support it, so the
chance for the next document is 5/9 and not 6/10. Statisticians call
this drawing without replacement, and the binomial formula does not hold
exactly. Nevertheless, as we discuss in this section, the binomial model
is enough for our purpose.

The model with finite evidence without replacement depends not only on
the initial proportion of evidence supporting \(H_1\), but on the size
of the pool of evidence. Let us call \(M\) the total number of pieces of
evidence in the world --- \(M\) being the size of the archive in our
previous example. In the following, we will treat \(M\) as known. As
before, \(H_1\) claims that most of those \(M\) pieces support it, and
the rival theory claims that not more than half the evidence does. If
\(K\) is the number of observations supporting \(H_1\), then \(H_1\)
claims that \(K>M/2\) and \(H_R\) claims that \(K\le M/2\). When \(M\)
is even, the archive split exactly in half, \(K=M/2\), belongs to the
rival, just as \(\theta = 1/2\) does in the binomial model. Therefore,
instead of a parameter \(\theta\) for the proportion of evidence, we
have the parameter \(K\).

For given values of \(K\) and \(M\), we can compute the exact
probability of obtaining \(N\) pieces of evidence, \(k\) of them
supporting \(H_1\):
\[p(E\mid M,K)= \dfrac{\displaystyle\binom{K}{k}\binom{M-K}{N-k}}{\displaystyle\binom{M}{N}}\]
This is called the hypergeometric formula. Even when \(M\) is known, we
do not know the specific number \(K\) of pieces of evidence supporting
\(H_1\) under each hypothesis. If the working hypothesis is true, we
know only that \(K>M/2\); we do not know by how much. We solve this just
as before. We can average the probability over all possible values of
\(K>M/2\), using uniform weights (or another distribution) for each
\(K\).

For the rival theory, we have the same two options as before: we could
treat it the same as \(H_1\) and average the probabilities over all
\(K \le M/2\), or we could choose the worst-case scenario: \(K\) equal
to the greatest integer less than or equal to \(M/2\). Hence, the
formulas for both Bayes factors are

\[\text{uniform-weights hypergeometric BF} = \frac{\displaystyle\frac{1}{c_1}\sum_{K>M/2} p(E\mid K,M)}{\displaystyle\frac{1}{c_R}\sum_{K\le M/2} p(E\mid K,M)},\]
\[\text{worst-case hypergeometric BF} = \frac{\displaystyle\frac{1}{c_1}\sum_{K>M/2} p(E\mid K,M)}{p(E\mid K=K_M,M)},\]
where \(c_1\) is the number of values of \(K\) greater than \(M/2\),
\(c_R\) is the number of values of \(K\) at or below \(M/2\), and
\(K_M\) is the largest \(K\) at or below \(M/2\). When \(M\) is odd,
\(c_1 = c_R = \frac{M+1}{2}\) and \(K_M = \frac{M-1}{2}\). When \(M\) is
even, \(c_1 = \frac{M}{2}\), \(c_R = \frac{M}{2}+1\), and
\(K_M = \frac{M}{2}\).

What remains is the choice of \(M\), since different values of \(M\)
change the resulting Bayes factors. Table~\ref{tbl-hyper} shows both
Bayes factors in our example, with 9 out of 12 observations favoring
\(H_1\) across different values of \(M\).

\begin{table}[H]

\caption{\label{tbl-hyper}Hypergeometric Bayes factors at 9 out of 12
observations supporting \(H_1\)}

\centering{

\centering
\begin{tabular}{lccc}
\toprule
Bayes factor & $(k, r)$ & $M$ & $BF$\\
\midrule
Uniform-weights & (9,3) & 20 & 162.88\\
Worst-case & (9,3) & 20 & 16.84\\
Uniform-weights & (9,3) & 100 & 25.22\\
Worst-case & (9,3) & 100 & 3.19\\
\bottomrule
\end{tabular}

}

\end{table}%

Moving from \(M=20\) to \(M=100\) strongly reduces both Bayes factors,
so the strength of our conclusion depends on the choice of \(M\).
However, as \(M\) increases, both Bayes factors computed here converge
to their respective binomial Bayes factors. In every case we computed in
which the worst-case binomial Bayes factor exceeds 1, the worst-case
hypergeometric Bayes factor is at least as large.\footnote{We computed
  every number of documents read from 1 to 40, every count favoring
  \(H_1\) of at least half of those read, and every archive size from
  the number read to 150 more than that.} The uniform-weights
hypergeometric Bayes factor can fall slightly below the binomial one
when \(M\) is even. With 3 of 5 documents favoring \(H_1\), read from an
archive of 6, it is 1.78 against a binomial 1.91. But in every case we
computed in which the binomial Bayes factors meet the rule of the main
paper, a uniform-weights Bayes factor of 20 or more and a worst-case
Bayes factor above 1, the hypergeometric Bayes factors meet it too.

In this way, in every case we computed, if the binomial Bayes factors
lead a researcher to set the rival aside, the hypergeometric Bayes
factors do too. Accordingly, there is no need to make a potentially
contested choice of \(M\), since every value of \(M\) we computed is
consistent with the binomial conclusion. However, if the binomial Bayes
factors are not sufficiently strong (for example, when the
uniform-weights Bayes factor is below 20) and the researcher wishes to
select a value of \(M\) --- for example, when the archive size is known
--- that yields a Bayes factor above 20, then \(M\) may be treated as a
sensitivity parameter: the researcher can examine how large \(M\) must
be for the uniform-weights Bayes factor to fall below 20, in the same
spirit as other parameters in Section~\ref{sec-sensitivity}.

\subsection{\texorpdfstring{Choosing the distributions of \(\theta\)
within each
range}{Choosing the distributions of \textbackslash theta within each range}}\label{sec-choosing-distributions}

Within each range, we have assigned \(\theta\) a uniform density,
\(f(\theta\mid H_1)=2\) for \(0.5<\theta<1\) and \(f(\theta\mid H_R)=2\)
for \(0<\theta\leq 0.5\), though other distributions are possible.

To choose a distribution, we need to clarify what it means for
\(\theta\) to be near 0.5 or near 0 and 1. When \(H_1\) is true,
\(\theta\) lies between 0.5 and 1. If \(\theta\) is near 0.5, a
researcher should find evidence supporting both theories, with only
slightly more favoring \(H_1\); if it is near 1, almost all the evidence
should favor \(H_1\). That is a substantive question about the world: in
a world where the working theory (or the rival) is right, how likely is
it to find evidence favoring the other?

We suggest rescaling the Beta\((a,b)\) family to obtain a distribution
on each range. The Beta\((a,b)\) distribution is defined for \(x\) in
\((0,1)\) with density function

\[f_X(x) = \frac{x^{a-1}(1-x)^{b-1}}{B(a,b)}, \qquad 0 < x < 1,\]

where \(B(a,b)=\displaystyle\int_{0}^1 t^{a-1}(1-t)^{b-1}\,dt\).

If we write \(\theta = 0.5 + 0.5X\), then \(X = 2\theta - 1\), and a
stretch of values for \(X\) becomes a stretch half as wide for
\(\theta\). Squeezing the same probability into half the width doubles
the density, so the density for \(\theta\) is twice the density for
\(X\) evaluated at \(2\theta - 1\):
\(2 f_X(2\theta - 1) = \frac{2}{B(a,b)}(2\theta-1)^{a-1}(2-2\theta)^{b-1}\).
Writing \((2-2\theta)^{b-1}\) as \(2^{b-1}(1-\theta)^{b-1}\) and
multiplying that \(2^{b-1}\) by the leading 2 gives \(2^{b}\), so the
density on \((0.5,1)\) is \[
f(\theta\mid a,b) = \frac{2^{b}}{B(a,b)}(2\theta-1)^{a-1}(1-\theta)^{b-1}, \qquad 0.5 < \theta < 1.
\] Conversely, if we write \(\theta = 0.5X\), then \(X = 2\theta\), the
factor is again 2, and \(2^{a-1}\) comes out of \((2\theta)^{a-1}\),
giving a distribution on \((0,0.5)\) with density function \[
f(\theta\mid a,b) = \frac{2^{a}}{B(a,b)}\theta^{a-1}(1-2\theta)^{b-1}, \qquad 0 < \theta < 0.5.
\] If \(a>b\), the distribution is concentrated on the right of the
range; if \(a<b\), it is concentrated on the left; and if \(a=b\), it is
symmetric about the middle of the range, with a shape that depends on
the value of \(a\): concentrated in the middle when \(a=b\) is above
one, uniform when \(a=b=1\), and rising toward both ends of the range
when \(a=b\) is below one, as Beta\((1/2, 1/2)\) does.

For example, a researcher who thinks that under each theory it is still
likely to find evidence supporting the other one could rescale a
Beta\((1/2,1)\) distribution for \(H_1\) and a Beta\((1,1/2)\)
distribution for \(H_R\), as shown in Figure~\ref{fig-beta1}. Both
densities pile up near 1/2. Conversely, a researcher who thinks that in
a world where one theory is right, there should be almost no evidence
supporting the other one could rescale a Beta\((1,1/2)\) distribution
for \(H_1\) and Beta\((1/2,1)\) for \(H_R\), as shown in
Figure~\ref{fig-beta2}: both densities pile up at the far ends.

\begin{figure}

\centering{

\pandocbounded{\includegraphics[keepaspectratio]{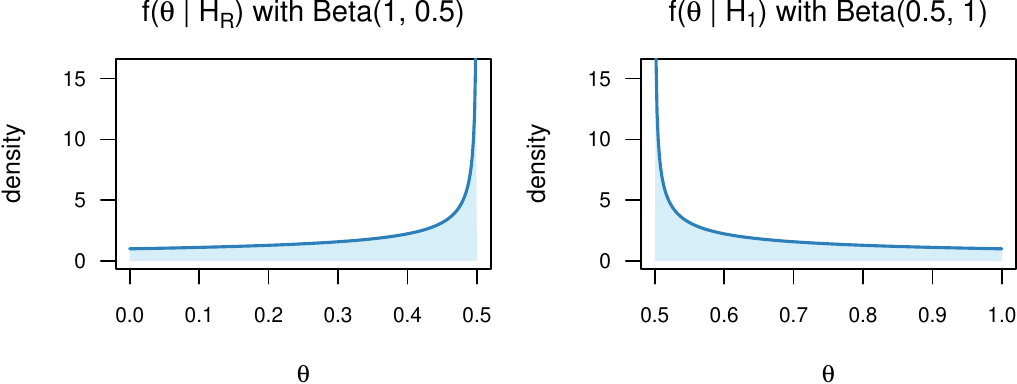}}

}

\caption{\label{fig-beta1}Rescaled Beta(1, 1/2) density for \(\theta\)
under \(H_R\) (left) and rescaled Beta(1/2, 1) density under \(H_1\)
(right): a researcher who expects evidence for both theories under
either.}

\end{figure}%

\begin{figure}

\centering{

\pandocbounded{\includegraphics[keepaspectratio]{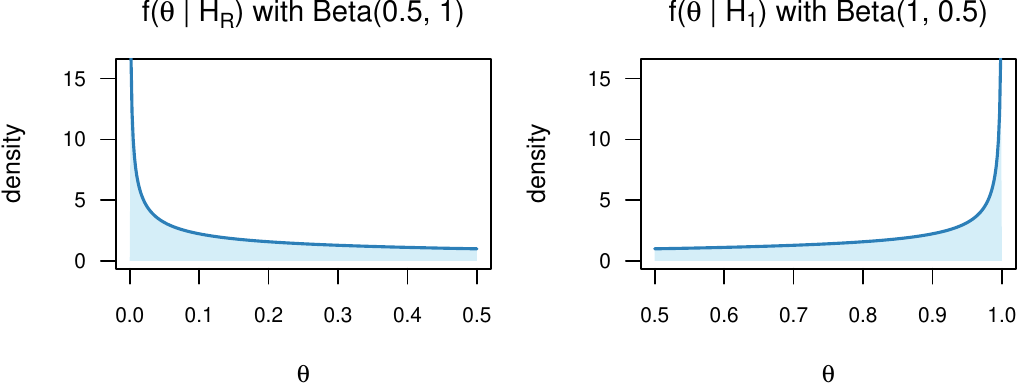}}

}

\caption{\label{fig-beta2}Rescaled Beta(1/2, 1) density for \(\theta\)
under \(H_R\) (left) and rescaled Beta(1, 1/2) density under \(H_1\)
(right): a researcher who expects almost no evidence for the other
theory under either.}

\end{figure}%

Each pair is a claim about the world, so a researcher who states one has
said something a reader can dispute, and the Bayes factor moves with the
choice. The researcher in country \countrylabel made 12 observations, 9
of which favored elite choice. Weights that expect evidence for both
theories in a world where either one holds (Figure~\ref{fig-beta1}) give
7.35. Weights that expect almost no evidence for the other theory
(Figure~\ref{fig-beta2}) give 28.08. The uniform-weights Bayes factor,
20.67, falls between them. Those two values fall on opposite sides of
the threshold of 20, so a researcher who expects evidence for both
theories would not set the weak-institutions theory aside on these
counts, and one who expects almost none for the other theory would.

Weights that expect almost no evidence for the other theory do not raise
the Bayes factor by making nine of twelve fit elite-choice theory
better. On the working theory's side, they do the opposite, and
splitting the ratio into the two averages shows why. Under the uniform
weights, the working theory's average is 0.14675 and the rival's is
0.00710. When the researcher expects almost no evidence for the other
theory, the working theory places its weight on values of \(\theta\)
near one, where nine of twelve is unlikely, and its average falls to
0.10630. The rival places its weight on values near zero, where nine of
twelve is more unlikely still, and the rival's average falls to 0.00379.
Both averages fit the counts worse than the uniform weights do. The
ratio rises because the rival's average falls by the larger factor.

\subsection{Properties of the Worst-Case Bayes
Factor}\label{sec-bounded}

The worst-case Bayes factor has two important properties. It is a lower
bound on the Bayes factor whenever the count supporting the working
theory is at least the count supporting the rival. For such counts,
whatever weights a proponent of the rival theory holds over the values
of \(\theta\) consistent with the rival's claim, the Bayes factor
against that proponent is at least the reported value
(Proposition~\ref{prp-bound}). And it comes with a bound on how often a
researcher who decides with it makes a mistake: a researcher who treats
\(t\) as the value at which to set the rival aside errs at most one time
in \(t\) when the rival is right, and the proof holds whether the
observations are independent draws or draws without replacement from a
finite collection of any size (Theorem~\ref{thm-error-rate}).

A Bayes factor comparing the two theories would divide the probability
of the counts averaged under the working theory's stated weights by the
probability of the counts averaged under whatever weights the rival's
proponent puts on the values of \(\theta\) consistent with the rival's
claim. We never learn the proponent's weights; critics of a case study
do not state priors. We can calculate a lower bound on the Bayes factor
without knowing the rival proponent's weights.

\begin{proposition}[The worst-case Bayes factor bounds the Bayes factor
from below]\protect\hypertarget{prp-bound}{}\label{prp-bound}

Let \(k \ge r\), let the researcher's weights be any probability
distribution stated before coding over \(\theta\) in \((1/2, 1]\), and
let the rival proponent's weights be any probability distribution over
\(\theta\) in \([0, 1/2]\). When the proponent's denominator is
positive, the Bayes factor --- the researcher's weighted numerator
divided by the proponent's weighted denominator --- is at least the
worst-case Bayes factor, that same numerator divided by the probability
of the counts at \(\theta = 1/2\).

\end{proposition}

\begin{proof}
When \(k \ge r\), the probability of the counts
\(\binom{N}{k}\theta^{k}(1-\theta)^{N-k}\), viewed as a function of
\(\theta\), rises on the whole interval \([0, 1/2]\): it increases up to
\(\theta = k/N\), and \(k \ge r\) puts \(k/N\) at or above one half. So
every value of \(\theta\) consistent with the rival's claim gives the
counts a probability at most the probability at one half. The
proponent's denominator is a weighted average of numbers each at most
that probability, so it is at most that probability, and dividing by it
gives a ratio at least the worst-case Bayes factor.
\end{proof}

We hold the observed counts, binomial observation model, and working
theory's weights fixed in this comparison. With uniform weights under
the working theory, the uniform-weights Bayes factor is one value
obtained by choosing the rival's weights; it is not an upper bound.

To see why no finite upper bound exists, suppose that \(k>0\) and that
the fixed numerator is positive. We can put a proportion \(c\) of the
rival's weight at \(\theta=1/2\) and the remaining weight at
\(\theta=0\), where \(0<c\leq 1\). With \(k>0\), the probability of the
observed counts at \(\theta=0\) is zero. The denominator is therefore
\(c\) times the denominator of the worst-case Bayes factor. Dividing by
it gives the worst-case Bayes factor divided by \(c\). As \(c\)
approaches zero, this ratio grows without bound. Together with
Proposition~\ref{prp-bound}, this establishes that, when \(k\geq r\) and
the stated positivity conditions hold, the possible finite Bayes factors
include every value from the worst-case Bayes factor upward.

\begin{theorem}[The error rate of the worst-case Bayes
factor]\protect\hypertarget{thm-error-rate}{}\label{thm-error-rate}

Fix the number of documents examined, \(N\). Let \(\theta_0\) with
\(0 < \theta_0 < 1\) be a value stated before coding, and let \(f\) be
any probability distribution over the values of \(\theta\) greater than
\(\theta_0\), also stated before coding. Write, for a count \(k\) of
documents supporting the working theory,

\[
\Lambda_{f, \theta_0}(k) \;=\;
\int \frac{\theta^{\,k}\,(1-\theta)^{\,N-k}}
          {\theta_0^{\,k}\,(1-\theta_0)^{\,N-k}} f(\theta)\, d\theta.
\]

Then, the average of \(\Lambda_{f, \theta_0}(k)\) across repetitions of
the observation process --- its average over the counts \(k\) the
process can produce, weighted by how often it produces each --- is at
most one whenever the share of the evidence supporting the working
theory is at most \(\theta_0\), both for \(N\) independent draws and for
\(N\) documents drawn at random, without replacement, from any finite
collection of \(M \ge N\) documents. Therefore, for every threshold
\(t > 0\), the researcher computes \(\Lambda_{f, \theta_0} \ge t\) with
probability at most \(1/t\).

\end{theorem}

\begin{proof}
Fix a value \(\theta > \theta_0\) in the support of \(f\), and write
\(\lambda = \theta/\theta_0\) and \(\mu = (1-\theta)/(1-\theta_0)\), so
that the integrand at \(\theta\) equals \(\mu^{N}\beta^{k}\) with
\(\beta = \lambda/\mu > 1\).

\emph{Step 1 (independent draws, and without replacement is at most with
replacement).} Under the model the count \(X\) of supporting
observations is binomial with success probability \(s\), the share the
evidence really has (we keep \(s\) and \(\theta\) apart from here on:
\(\theta\) ranges over the values the weights \(f\) cover, and \(s\) is
the one share the evidence really has), and each of the \(N\)
independent draws contributes a factor of \(1 - s + s\beta\) to the
expectation, so \(\mathrm{E}[\beta^{X}] = (1 - s + s\beta)^{N}\)
exactly. Under draws without replacement, \(X\) is the sum of \(N\)
draws without replacement from a population of \(M\) values, \(K_1\)
ones and \(M - K_1\) zeros, with \(s = K_1/M\). The map
\(x \mapsto \beta^{x}\) is continuous and convex, and Hoeffding (1963,
sec. 6, Theorem 4) states that for any continuous convex function, the
expectation of the function of such a sum is at most its expectation
under draws with replacement from the same population, which is the
binomial case just computed. So in both cases
\(\mathrm{E}[\beta^{X}] \le (1 - s + s\beta)^{N}\).

\emph{Step 2 (the binomial bound is at most one exactly on the rival's
claim).} Multiplying by \(\mu^{N}\) gives
\(\mathrm{E}[\mu^{N}\beta^{X}] \le (\mu + s(\lambda - \mu))^{N}\). Since
\(\theta > \theta_0\) makes \(\lambda > 1 > \mu\), the base is at most
one exactly when \(s \le (1-\mu)/(\lambda-\mu)\), and that quantity
simplifies to \(\theta_0\): \[
\frac{1-\mu}{\lambda-\mu}
= \frac{(\theta - \theta_0)/(1-\theta_0)}{(\theta - \theta_0)/\bigl(\theta_0(1-\theta_0)\bigr)}
= \theta_0 .
\]

So the expectation of the integrand at \(\theta\) is at most one
whenever \(s \le \theta_0\), which is the rival's claim itself.

\emph{Step 3 (the endpoint, and the average over \(f\)).} Steps 1 and 2
cover every \(\theta\) in \((\theta_0, 1)\). At the endpoint
\(\theta = 1\) the integrand is \(\theta_0^{-N}\) at \(k = N\) and zero
elsewhere, and \(\beta\) is not finite, so that case needs its own line:
the probability of a unanimous sample is \(s^{N}\) under independent
draws and \(\prod_{i=0}^{N-1}(K_1 - i)/(M - i) \le s^{N}\) without
replacement, since each factor is at most \(K_1/M\), so the expectation
there is at most \(s^{N}\theta_0^{-N} \le 1\) on the rival's claim. The
integrand is nonnegative, and the sum over the \(N + 1\) counts is
finite, so the integral over \(\theta\) and the average over counts may
be exchanged, giving an average of at most
\(\int 1 \, \mathrm{d}f = 1\). The probability statement follows from
Markov's inequality: a list of nonnegative numbers whose average is at
most 1 has entries of \(t\) or more in at most a \(1/t\)-fraction of its
places.
\end{proof}

Setting \(\theta_0 = 1/2\) and taking \(f\) uniform on the values of
\(\theta\) above one half, we obtain the property for the worst-case
Bayes factor.

\section{Sensitivity Analysis}\label{sec-sensitivity}

In a sensitivity analysis, the researcher asks how much an assumption
behind the Bayes factor would have to fail before the conclusion
changed.

We treat five sensitivity questions below, each in its own subsection:
\textbf{coding} (how many observations would need to be re-coded to
change what the researcher reports?), \textbf{observation bias} (how
strongly would the researcher's search have to have favored one theory's
evidence over the other's?), \textbf{probative weight} (how much extra
weight would a smoking-gun observation have to carry?), the
\textbf{weights on \(\theta\)} in the uniform-weights Bayes factor (how
does the uniform-weights Bayes factor change under other weights on the
values of \(\theta\) within each hypothesis's range?), and
\textbf{dependence among observations} (how much would treating
correlated observations as independent overstate the evidence?). For
each question, we recompute both Bayes factors and examine how the
conclusion shifts.

\subsection{Sensitivity to coding error}\label{sec-sens-coding}

A peer cannot re-read every observation. The published paper typically
includes a quote to illustrate why one document supports \(H_1\), but a
peer cannot inspect every document as the original researcher did.
Suppose a peer now presses our researcher on how the evidence was coded.
The coding of any single piece of evidence can be disputed but not
settled at a distance; what we can compute is how many observations a
peer would have to re-code for the conclusion to change.

\subsubsection{The uniform-weights Bayes factor's tipping point under
re-coding}\label{the-uniform-weights-bayes-factors-tipping-point-under-re-coding}

We want to determine how many of the \(k\) pro-\(H_1\) observations,
call this number \(x\), a peer would have to re-code as favoring \(H_R\)
before the uniform-weights Bayes factor falls to \(20\) or the
worst-case Bayes factor falls to 1. Re-coding does not change how many
observations there are: still \(N = k + r\) of them, now \(k - x\)
favoring \(H_1\) and \(r + x\) favoring \(H_R\). The Bayes factors at
those counts are

\[
\text{uniform-weights BF} = \dfrac{\displaystyle\int_{0.5}^1 \theta^{k-x}(1-\theta)^{r+x}\,d\theta}{\displaystyle\int_{0}^{0.5} \theta^{k-x}(1-\theta)^{r+x}\,d\theta}, \qquad \text{worst-case BF} = \dfrac{2\displaystyle\int_{0.5}^1 \theta^{k-x}(1-\theta)^{r+x}\,d\theta}{(\frac{1}{2})^N}.
\]

The tipping point occurs when \(\text{uniform-weights BF}=20\) or
\(\text{worst-case BF}=1\), and \(x\) is the minimum value that solves
either of those equations. Re-codings come in whole observations, so
what we report is the smallest whole number of re-codings at which the
conclusion stops supporting \(H_1\). If that number is larger than any
careful reader of these sources would propose, the conclusion does not
turn on the coding of individual observations.

If we apply this sensitivity analysis to our running example, where
\(k = 9\) observations support \(H_1\) and \(r = 3\) support \(H_R\),
the analysis shows that re-coding just 1 of the 9 pro-\(H_1\)
observations takes the uniform-weights Bayes factor below 20 (the
equation above is solved at \(x = 0.026\), less than a single
observation, which is why the whole-number answer is one). The
conclusion tolerates so little re-coding because the uniform-weights
Bayes factor sits just above the threshold to begin with.

\subsection{Sensitivity to observation bias}\label{sec-sens-bias}

A peer cannot easily verify the conditions under which the researcher's
evidence was gathered: the archive may have been curated, the interviews
mediated by gatekeepers, or the researcher may have followed leads
toward one side more than the other. Each is a path to \emph{observation
bias} --- a systematic tendency to find evidence on one side more
readily than the other. We parameterize this bias by a multiplicative
factor \(\omega\) on the odds of observing pro-\(H_1\) versus
pro-\(H_R\) evidence and ask how strong \(\omega\) would have to be to
overturn the conclusion.

\subsubsection{The binomial under a tilted
likelihood}\label{the-binomial-under-a-tilted-likelihood}

Bias enters the model as a multiplicative tilt on the per-observation
odds. Without bias, each observation supports \(H_1\) with probability
\(\theta\) and \(H_R\) with probability \(1 - \theta\). Under bias
factor \(\omega\), pro-\(H_1\) evidence is \(\omega\) times more likely
to be found than it would be otherwise, so the researcher observes
pro-\(H_1\) evidence with a new probability \(q\) that satisfies

\[\dfrac{q}{1-q}=\omega\cdot \dfrac{\theta}{1-\theta}.\]

Solving gives \(q = \omega\theta / (1-\theta+\omega\theta)\). The
likelihood of observing \(k\) pro-\(H_1\) out of \(N\) observations is
then binomial with probability \(q\):

\[p(E \mid \theta,\omega)=\binom{N}{k}q^k(1-q)^{N-k}.\]

The parameters of interest remain \(\theta\) and \(\omega\); \(q\) is
shorthand for the function of both written above. The Bayes factors at a
given \(\omega\) are
\[\text{uniform-weights BF}(\omega)=\dfrac{\displaystyle\int_{0.5}^1q^{k}(1-q)^{r}\,d\theta}{\displaystyle\int_{0}^{0.5}q^{k}(1-q)^{r}\,d\theta},\qquad \text{worst-case BF}(\omega)=\dfrac{2\displaystyle\int_{0.5}^1q^{k}(1-q)^{r}\,d\theta}{(q^{*})^{k}(1-q^{*})^{r}},\]
where \(q^{*}\) is the value of \(q\) at which we evaluate the rival.
When \(k \ge r\),
\(q^{*}=\min\bigl(\frac{\omega}{1+\omega}, \frac{k}{N}\bigr)\). Of the
values of \(q\) that the rival's values of \(\theta\) produce, this is
the one at which the counts are most probable
(Section~\ref{sec-bias-rival-share} explains why). When \(k < r\) we
hold the rival at \(\theta = \frac{1}{2}\), as without bias, so
\(q^{*} = \frac{\omega}{1+\omega}\). We solve
\(\text{uniform-weights BF}(\omega)=20\) and
\(\text{worst-case BF}(\omega)=1\) for \(\omega\) numerically and choose
the minimum of both.

For the running example (\(k = 9\) pro-\(H_1\) out of \(N = 12\)), the
uniform-weights Bayes factor falls to 20 at \(\omega = 1.01\) ---
essentially no bias at all. Because the uniform-weights Bayes factor is
20.67 at these counts, just above the threshold, the smallest tilt
toward finding pro-\(H_1\) evidence is enough to take it below 20.

\subsubsection{\texorpdfstring{Which value of \(\theta\) the rival gets
under a biased
search}{Which value of \textbackslash theta the rival gets under a biased search}}\label{sec-bias-rival-share}

The worst-case Bayes factor divides the probability of the counts under
\(H_1\), averaged over the values of \(\theta\) above one half, by the
probability of the counts at a single value of \(\theta\) at or below
one half. The proponent of \(H_R\) chooses that single value so as to
make the Bayes factor smallest, which means choosing the value at or
below one half at which the counts are most probable. Without bias, the
probability of \(k\) of \(N\) at \(\theta\) is
\(\binom{N}{k}\theta^k(1-\theta)^{r}\), which rises up to
\(\theta = k/N\) and falls after. When \(k \ge r\), \(k/N\) is at or
above one half, so among the values at or below one half the probability
is largest when \(\theta = \frac{1}{2}\), and the proponent of \(H_R\)
chooses one half.

Bias changes the probability of the counts at each \(\theta\). The
probability becomes \(\binom{N}{k}q^k(1-q)^r\) with
\(q = \omega\theta/(\omega\theta + 1 - \theta)\) in place of \(\theta\).
This is the same formula with a different input, so the probability is
largest when \(q = k/N\). The odds of \(q\) are \(\omega\) times the
odds of \(\theta\), so the value of \(\theta\) at which \(q\) equals
\(k/N\) has odds \(\theta/(1-\theta) = (k/r)/\omega\), which gives
\(\theta = k/(k + \omega r)\). That value is below one half exactly when
\(\omega\) is above \(k/r\). So the proponent of \(H_R\) chooses
\(\theta = \frac{1}{2}\) when \(\omega \le k/r\) and
\(\theta = k/(k + \omega r)\) when \(\omega > k/r\). Written in terms of
\(q\), the denominator is evaluated at
\(q^{*} = \min\bigl(\frac{\omega}{1+\omega}, \frac{k}{N}\bigr)\), the
formula in the display above.

Andersen's study (Section~\ref{sec-app-andersen}) shows the difference
in numbers. Its counts are nine of twelve, so \(k/r = 3\), and its
sensitivity table reports the worst-case Bayes factor at \(\omega = 5\).
When \(\theta = \frac{1}{2}\), a search with \(\omega = 5\) gives
\(q = 5/6\), and the probability of nine of twelve is 0.197. When
\(\theta = 0.375\), which is \(9/(9 + 5 \times 3)\), the same search
gives \(q = 0.75 = 9/12\), and the probability is 0.258. The numerator,
the average of the probability of nine of twelve over the values of
\(\theta\) above one half, each at its own \(q\), is 0.0554. With the
rival held at one half, the worst-case Bayes factor would be
\(0.0554/0.197 = 0.28\). At \(\theta = 0.375\) it is
\(0.0554/0.258 = 0.21\), the value the table reports. The two
computations differ in five entries of the per-study sensitivity tables,
the entries at which \(\omega\) is above the study's \(k/r\): Andersen's
(\(k/r = 3\)), Mor's (\(k/r = 4\)), and Winward's (\(k/r = 4.67\)) at
\(\omega = 5\), and Pavone and Stiansen's (\(k/r = 1.75\)) at
\(\omega = 2\) and \(\omega = 5\).

Evaluating the rival at \(k/(k + \omega r)\) rather than at one half
changes no decision and no tipping point. All five entries are below 1
under either computation, and the rule sets the rival aside only when
the worst-case Bayes factor is above 1, so a researcher following the
rule keeps the rival in each of the five cases either way. In every
study whose worst-case Bayes factor is above 1 without bias, that Bayes
factor falls to 1 at a bias below \(k/r\): at \(\omega = 1.56\) against
\(k/r = 3\) for Andersen, 2.01 against 4 for Mor, 2.49 against 4.67 for
Winward, and 2.81 against 5.5 for Hammoud-Gallego and Freier. Below
\(k/r\) the two computations agree, so those biases are unchanged. The
same holds for any counts with \(k > r\). At \(\omega = k/r\), the
search turns \(\theta = \frac{1}{2}\) into \(q = k/N\), so the
denominator is the largest probability the counts can have at any \(q\),
while every \(\theta\) above one half gives a \(q\) above \(k/N\) and so
a smaller probability. The numerator averages those smaller
probabilities, so the worst-case Bayes factor is below 1 at
\(\omega = k/r\). Without bias it is above 1 whenever the rule holds, so
it reaches 1 at some bias between 1 and \(k/r\), where the rival's value
is still one half.

\subsection{Sensitivity to probative weight}\label{sec-sens-weights}

Process-tracing scholars often distinguish between observations that are
necessary to entertain \(H_1\) (hoop tests) and decisive observations
(smoking guns), which carry different probative value. Fairfield and
Charman (2022, 130--38) formalize this idea through decibel ratings,
which translate into multiplicative factors on the likelihood ratio. The
framework so far has set \(w_i = 1\) for every observation: each piece
of evidence carries equal probative force. This is a baseline
assumption, not a claim about how evidence actually behaves, and we now
ask what changes if a researcher meaningfully believes that one or two
observations carry more weight than the rest.

The mechanism is simple. As the main paper shows, an observation of
integer weight \(w_i\) is treated as \(w_i\) identical unit-weight
observations. The weighted totals are \[
W = \sum_{i \in H_1} w_i, \qquad R = \sum_{j \in H_R} w_j,
\] and each Bayes factor is recomputed at \((W, R)\) in place of
\((k, r)\). Both are Bayes factors based on the weighted counts.

As the main paper shows, 8 ordinary pro-\(H_1\) observations, 3 ordinary
pro-\(H_R\) observations and one pro-\(H_1\) observation with a weight
of 10 give \(W = 18\) and \(R = 3\) and move both Bayes factors above
20. The sensitivity question runs in the other direction. Suppose a peer
reviewer challenges the analysis from the other side, identifying one of
the three pro-\(H_R\) observations --- say a memo from a centrist
legislator warning about authoritarian risks --- as a smoking gun for
the rival theory. How much weight would that reviewer need to assign to
this single observation to change the researcher's conclusion?

In this example, no weighting is required. The uniform-weights Bayes
factor already exceeds 20, and the worst-case Bayes factor already
exceeds 1. However, in Section~\ref{sec-applications}, we discuss an
example in which a smoking gun is assigned a weight of 10, while the
unweighted evidence is insufficient to support \(H_1\). In that case,
the central question is how much weight is actually necessary --- and
whether it is reasonable to assign such a weight to the smoking gun.

\subsection{\texorpdfstring{Sensitivity to the weights on \(\theta\)
(uniform-weights Bayes factor
only)}{Sensitivity to the weights on \textbackslash theta (uniform-weights Bayes factor only)}}\label{sec-sensitivity-prior}

The uniform-weights Bayes factor depends on the weights placed on
\(\theta\) within each theory's range. So far, we have used the uniform
on each range, which assigns equal density to every value of \(\theta\)
between 0.5 and 1 under \(H_1\) and between 0 and 0.5 under \(H_R\). A
reviewer might reasonably ask how the conclusion changes under different
weights. The question is asked with a rescaled Beta pair, as
Section~\ref{sec-choosing-distributions} describes, rather than with a
single density over the entire interval from zero to one. A single
density sets two things at once: how much weight each theory puts on the
values of \(\theta\) inside its own range, and how probable the two
theories are relative to each other before any evidence arrives. The
second is the reader's prior odds, which the paper leaves to the reader;
the ratio of posterior masses computed under such a density is posterior
odds, not a Bayes factor. Under a rescaled pair, each half carries a
density of its own that integrates to one over that half, so the pair
says how each theory spreads its weight and says nothing about the prior
odds on the theories, which stay whatever the reader holds.

\begin{longtable}[]{@{}
  >{\raggedright\arraybackslash}p{(\linewidth - 6\tabcolsep) * \real{0.4242}}
  >{\centering\arraybackslash}p{(\linewidth - 6\tabcolsep) * \real{0.2348}}
  >{\centering\arraybackslash}p{(\linewidth - 6\tabcolsep) * \real{0.2348}}
  >{\centering\arraybackslash}p{(\linewidth - 6\tabcolsep) * \real{0.1061}}@{}}

\caption{\label{tbl-sens-distributions}The uniform-weights Bayes factor
for \(k = 9\) pro-\(H_1\) out of \(N = 12\) under five pairs of rescaled
Beta distributions, ordered by the Bayes factor. Each pair leaves the
prior odds on the two theories untouched. Two of the five fall below the
threshold of 20.}

\tabularnewline

\toprule\noalign{}
\begin{minipage}[b]{\linewidth}\raggedright
What the researcher expects to find
\end{minipage} & \begin{minipage}[b]{\linewidth}\centering
Weights for \(\theta \mid H_1\)
\end{minipage} & \begin{minipage}[b]{\linewidth}\centering
Weights for \(\theta \mid H_R\)
\end{minipage} & \begin{minipage}[b]{\linewidth}\centering
Bayes factor
\end{minipage} \\
\midrule\noalign{}
\endhead
\bottomrule\noalign{}
\endlastfoot
Evidence for both theories under either & Beta(1/2, 1) & Beta(1, 1/2) &
7.35 \\
More weight at both ends of each half & Beta(1/2, 1/2) & Beta(1/2, 1/2)
& 9.53 \\
No value of \(\theta\) favored over another (the default) & Beta(1, 1) &
Beta(1, 1) & 20.67 \\
Almost no evidence for the other theory & Beta(1, 1/2) & Beta(1/2, 1) &
28.08 \\
Each theory's \(\theta\) near the middle of its half & Beta(10, 10) &
Beta(10, 10) & 252.92 \\

\end{longtable}

The five values span 7.35 to 252.92, and two of the five leave the
counts short of the threshold of 20. The two pairs that state an
expectation about finding evidence for the other theory are the ones a
researcher is most likely to argue about, because each states an
ordinary belief about the case. Expecting to find some evidence for weak
institutions even in a world where elites drove the collapse pulls both
theories' weights toward one half, and the Bayes factor falls to 7.35.
Held more strongly, that belief drives the Bayes factor toward one: if
both theories placed all their weight at \(\theta = 1/2\), they would
make nine of twelve exactly as probable as each other, and the evidence
would not tell them apart at all. Expecting almost no evidence for the
other theory instead pushes the weights to the far ends, and the Bayes
factor rises to 28.08. The last row is the strongest of the five and the
least modest: it has each theory predicting a value of \(\theta\) near
the middle of its own half, three quarters of the evidence for the
theory that holds, which makes nine of twelve much more probable under
elite choice than under weak institutions. What the table hands a reader
is the question behind the number: before seeing these twelve
observations, how lopsided did the researcher expect a world governed by
either theory to look?

\subsection{Sensitivity to dependence among
observations}\label{sec-sens-dependence}

The model treats each observation as an independent draw, and the
worst-case Bayes factor's error rate also covers draws without
replacement from one finite collection. Neither covers the further
dependence qualitative evidence often carries: several observations from
the same interview, the same event reported in different sources,
documents quoting one another. Treating such observations as separate
draws overstates the amount of evidence the researcher has. This
subsection reports how that overstatement affects the uniform-weights
Bayes factor, where an explicit dependence model can be computed
exactly, and openly states what remains open for the worst-case Bayes
factor.

\subsubsection{The uniform-weights Bayes factor under exchangeable
dependence}\label{the-uniform-weights-bayes-factor-under-exchangeable-dependence}

We model dependence the way survey samplers model clustered responses:
given \(\theta\), the count of supporting observations follows a
beta-binomial distribution in which each observation supports the
working theory with a probability whose mean is \(\theta\), and each
pair of observations correlates at \(\tau\), the within-study
correlation, with \(\tau = 0\) recovering independence. The Bayes factor
under this model integrates the beta-binomial likelihood over each
hypothesis's region of \(\theta\), exactly as the independent-draws
version integrates the binomial likelihood, and this exact computation
is the reported analysis.

{\def\LTcaptype{none} 
\begin{longtable}[]{@{}
  >{\centering\arraybackslash}p{(\linewidth - 6\tabcolsep) * \real{0.2041}}
  >{\centering\arraybackslash}p{(\linewidth - 6\tabcolsep) * \real{0.2449}}
  >{\centering\arraybackslash}p{(\linewidth - 6\tabcolsep) * \real{0.3163}}
  >{\centering\arraybackslash}p{(\linewidth - 6\tabcolsep) * \real{0.2347}}@{}}
\toprule\noalign{}
\begin{minipage}[b]{\linewidth}\centering
Correlation \(\tau\)
\end{minipage} & \begin{minipage}[b]{\linewidth}\centering
Exact beta-binomial BF
\end{minipage} & \begin{minipage}[b]{\linewidth}\centering
Count-rescaling approximation
\end{minipage} & \begin{minipage}[b]{\linewidth}\centering
Approximation / exact
\end{minipage} \\
\midrule\noalign{}
\endhead
\bottomrule\noalign{}
\endlastfoot
0.02 & 13.78 & 14.20 & 1.030 \\
0.05 & 8.93 & 9.64 & 1.080 \\
0.10 & 5.51 & 6.41 & 1.164 \\
0.20 & 3.10 & 4.07 & 1.313 \\
0.30 & 2.21 & 3.15 & 1.428 \\
\end{longtable}
}

The independent-draws Bayes factor at the running example is 20.67; the
exact beta-binomial value is 13.78 at a correlation of only
\(\tau = 0.02\) and 2.21 at \(\tau = 0.30\), so even mild clustering
takes the conclusion below the threshold of 20. A researcher who
suspects the observations cluster --- and who can state a correlation
--- reports the exact column. A reader can also invert the table and ask
how much correlation the conclusion survives, and at these counts the
answer is: almost none. The table's third column is the shortcut a
reader may have met elsewhere --- divide the counts by the design effect
\(D = 1 + (N-1)\tau\) and recompute --- shown with its error: the
shortcut overstates the exact value by a factor that grows to 1.43 at
\(\tau = 0.30\), so we use it only as an interpretation aid, never as
the reported analysis.

\subsubsection{Dependence and the frequency of conclusive-looking
counts}\label{dependence-and-the-frequency-of-conclusive-looking-counts}

Dependence also changes how often an evenly split evidence base produces
counts that look conclusive. With twelve independent observations and an
evenly split evidence base, the count reaches nine or more --- the
running example's coding --- 7.3 percent of the time. At a within-study
correlation of 0.10, that frequency is 17.3 percent, and at 0.30 it is
29.3 percent: correlated observations pile up on one side more easily,
so lopsided counts alone mean less.

A reader wanting a single summary of how much dependence discounts a
sample can use the design effect \(D\) of Kish (1965) --- the factor by
which clustering inflates the variance of a mean, so that \(N\)
correlated observations carry the information of roughly \(N/D\)
independent ones --- but the reported analysis above is the exact
beta-binomial computation, and the design-effect rescaling is an
interpretation aid only.

\subsubsection{What remains open for the worst-case Bayes
factor}\label{what-remains-open-for-the-worst-case-bayes-factor}

Theorem~\ref{thm-error-rate} covers one form of dependence: when
documents are drawn at random, without replacement, from a finite
collection, the worst-case Bayes factor reaches any value \(t\) at most
one time in \(t\), just as with independent draws. Positive dependence
of the clustered kind above is another matter: nothing yet says what a
cluster of correlated documents does to the worst-case Bayes factor's
error rate, and we do not report a dependence-adjusted worst-case value.
A model that builds dependence into the observation process directly is
the subject of a separate paper.

\section{Applications}\label{sec-applications}

The subsections that follow apply the framework to six recent
process-tracing studies. We begin with how we selected them, then take
each paper in turn, and close with what the set as a whole teaches us
about the framework.

Table~\ref{tbl-app-all} summarizes the Bayes factor values and the
results for the observation bias and re-coding sensitivity analysis.

\begin{longtable}[]{@{}
  >{\raggedright\arraybackslash}p{(\linewidth - 10\tabcolsep) * \real{0.3265}}
  >{\centering\arraybackslash}p{(\linewidth - 10\tabcolsep) * \real{0.1020}}
  >{\centering\arraybackslash}p{(\linewidth - 10\tabcolsep) * \real{0.2041}}
  >{\centering\arraybackslash}p{(\linewidth - 10\tabcolsep) * \real{0.1531}}
  >{\centering\arraybackslash}p{(\linewidth - 10\tabcolsep) * \real{0.1429}}
  >{\centering\arraybackslash}p{(\linewidth - 10\tabcolsep) * \real{0.0714}}@{}}

\caption{\label{tbl-app-all}The six studies under both Bayes factors and
the sensitivity results for bias (\(\omega^*\)) and re-coding (\(x^*\)).
Counts are the two-coder consensus codings, except that we re-coded one
of Hammoud-Gallego and Freier's observations from the article
(Section~\ref{sec-app-hammoud-gallego-freier}).}

\tabularnewline

\toprule\noalign{}
\begin{minipage}[b]{\linewidth}\raggedright
Paper
\end{minipage} & \begin{minipage}[b]{\linewidth}\centering
\((k, r)\)
\end{minipage} & \begin{minipage}[b]{\linewidth}\centering
Uniform-weights BF
\end{minipage} & \begin{minipage}[b]{\linewidth}\centering
Worst-case BF
\end{minipage} & \begin{minipage}[b]{\linewidth}\centering
\(\omega^{*}\)
\end{minipage} & \begin{minipage}[b]{\linewidth}\centering
\(x^*\)
\end{minipage} \\
\midrule\noalign{}
\endhead
\bottomrule\noalign{}
\endlastfoot
Steinsson 2024 & (12, 0) & 8,191 & 630 & 5.49 & 4 \\
Winward 2021 & (14, 3) & 264 & 21.34 & 1.74 & 2 \\
Hammoud-Gallego and Freier 2023 & (11, 2) & 154 & 14.91 & 1.70 & 2 \\
Mor 2022 & (8, 2) & 29.57 & 4.00 & 1.14 & 1 \\
Andersen 2024 & (9, 3) & 20.67 & 2.73 & 1.01 & 1 \\
Pavone and Stiansen 2022 & (7, 4) & 4.16 & 0.83 & \(-\) & 0 \\

\end{longtable}

In the per-study sensitivity tables below, the observation-bias axis
shows the Bayes factors for different possible values of the bias factor
\(\omega\). On that axis we compute the worst-case Bayes factor with the
rival at the value of \(\theta\) that makes the counts most probable
under the biased search. When \(\omega\) is larger than \(k/r\), that
value is below one half (Section~\ref{sec-bias-rival-share}). On the
axis for the weights on \(\theta\), each pair is a set of weights the
researcher might state within each theory's range
(Section~\ref{sec-choosing-distributions}): the uniform-weights Bayes
factor takes the pair on both theories' ranges, and the worst-case Bayes
factor takes it on the working theory's range only
(Section~\ref{sec-sensitivity-prior}). Weighted values on the
smoking-gun axis are the two Bayes factors at the weighted counts
(Section~\ref{sec-sens-weights}).

\subsection{How we selected the
corpus}\label{how-we-selected-the-corpus}

We queried OpenAlex for articles published from 2020 onward in six
political science journals (the \emph{American Political Science
Review}, the \emph{American Journal of Political Science}, the
\emph{Journal of Politics}, \emph{Comparative Political Studies},
\emph{World Politics}, and \emph{Perspectives on Politics}). We applied
two complementary filters to each journal: a keyword filter matching any
of the standard process-tracing vocabulary (\texttt{process\ tracing},
\texttt{hoop\ test}, \texttt{smoking\ gun}, \texttt{doubly\ decisive},
\texttt{straw\ in\ the\ wind}, \texttt{causal-process\ observation},
\texttt{set-theoretic}) in the title or abstract, and a citation filter
capturing every article that cites any of seven anchor texts on
process-tracing methodology.\footnote{The seven anchor texts are:
  Bennett and Checkel (eds., 2015, \emph{Process Tracing: From Metaphor
  to Analytic Tool}); Beach and Pedersen (2018, \emph{Process-Tracing
  Methods}; 2016, \emph{Causal Case Study Methods}); Collier (2011,
  \emph{Understanding Process Tracing}); Mahoney (2012, \emph{Logic of
  Process-Tracing Tests}); Fairfield and Charman (2022, \emph{Social
  Inquiry and Bayesian Inference}); Humphreys and Jacobs (2023,
  \emph{Integrated Inferences}).} The query returned 50 candidate papers
(search code at \texttt{process\_tracing\_corpus/build\_corpus.py}; raw
output at \texttt{process\_tracing\_corpus/papers.csv}). We reviewed the
50 candidates by hand against three further criteria --- discrete
observations that can be counted, a named working theory, and at least
one rival, and substantive diversity --- and shortlisted ten papers
(\texttt{process\_tracing\_corpus/shortlist.md}). Three of the ten fell
outside the scope of this paper, and we set them aside: Coppock and Kaur
(2022) is a meta-analysis of truth-commission cases rather than a
process-tracing study; Slaven et al. (2020) tests four theories against
one another with none of them the working theory, so its interviews are
coded for which of the four each supports rather than for or against a
single theory, which is what our counts need; Goertz and Haggard (2023)
introduces Large-N Qualitative Analysis as a methodology rather than
presenting an empirical process trace. The \emph{Journal of Politics}
returned zero candidate papers under either filter; we note this absence
without further interpretation. We applied our framework to six of the
seven remaining papers --- the six that the main paper summarizes ---
which are developed in detail in the subsections below. We set the
seventh aside: Leipziger, Rorbaek, and Skaaning's (2025) fifteen-case
comparison is the one large-N qualitative design in the set, and the
question of how to count and weight cases in large-N qualitative
analysis raises questions we do not take up here.

\subsection{Mor (2022): voter-driven coalition formation in
nineteenth-century Prussia}\label{sec-app-mor}

Mor (2022) asks how the ethnic dimension of party competition emerged in
nineteenth-century Prussia. Mor's working theory (\(H_1\)) traces
ethnic-party formation to a voter-driven sequence: government policies
aggrieve a religious or ethnic group, voters then coordinate on ethnic
identity, and political entrepreneurs respond tactically. The
conventional rival (\(H_R\)) traces it to elite initiatives: clerical
authority, mass organizations, and entrepreneurial coalition-building
drive the change from the top down. The design is a four-case sequence
within Prussia (1848--49, 1852, 1855--66, 1867--71), chosen to hold the
entrepreneur infrastructure approximately constant while policy
aggrievement varies.

We coded 24 observations from Mor's article into the coding table in
\texttt{replications/mor\_2022/cases.csv}, one observation per line,
with a page citation and a verbatim quote. Both coders agreed that 8
observations favor \(H_1\) and that 2 favor \(H_R\); we set aside the
remaining 14 as ambiguous. The agreed pro-\(H_1\) observations include
the finding from Mor's first case that massive entrepreneurial effort in
1848--49 failed to consolidate a Catholic bloc under neutral policy
(p.140); the immediate emergence of a Catholic coalition after the 1852
Raumer Decrees, despite weaker entrepreneurial capacity than in 1848
(p.141); the dealignment of 1855--66 as policy returned to neutral
(p.142); the November 1867 Reichstag election, in which
Catholic-majority constituencies elected anti-Bismarck liberals despite
a Church ban on clerical political activity (p.144); and the
constituency-level statistical analysis showing that proxies for Church
organizational capacity do not predict the change in liberal vote share
(p.152). The two agreed pro-\(H_R\) observations are the
entrepreneur-driven sub-coalition in Rhineland-Westphalia in 1848
(p.140) and the 1870 relaunch of the Zentrum, in which previous leaders
codified an opposition program and built coordination across
associations (p.144). Many of the 14 ambiguous observations concern two
episodes that fit both the voter theory and the entrepreneur rival: the
February 1867 ballot-tearing and the sequencing claim about 1870.

At \((k, r) = (8, 2)\), the uniform-weights Bayes factor is 29.57, above
the threshold of 20, and the worst-case Bayes factor is 4.00, above 1.
This indicates that the evidence is substantially more consistent with
\(H_1\) than with \(H_R\). Even when the elite theory is given the
benefit of the doubt, the evidence remains more consistent with \(H_1\).
Two of the eight agreed pro-\(H_1\) observations are results of Mor's
constituency-level statistical analysis. Without them, the counts are
six-and-two, and the uniform-weights Bayes factor is 10.1, below the
threshold.

\begin{longtable}[]{@{}
  >{\raggedright\arraybackslash}p{(\linewidth - 4\tabcolsep) * \real{0.6866}}
  >{\centering\arraybackslash}p{(\linewidth - 4\tabcolsep) * \real{0.1567}}
  >{\centering\arraybackslash}p{(\linewidth - 4\tabcolsep) * \real{0.1567}}@{}}

\caption{\label{tbl-app-mor-sens}Bayes factors for Mor (2022) when
coding error, observation bias, a smoking-gun weight, and the weights on
\(\theta\) each take three values. Agreed counts \((k, r) = (8, 2)\); 14
observations set aside as ambiguous.}

\tabularnewline

\toprule\noalign{}
\begin{minipage}[b]{\linewidth}\raggedright
What we vary
\end{minipage} & \begin{minipage}[b]{\linewidth}\centering
Uniform-weights BF
\end{minipage} & \begin{minipage}[b]{\linewidth}\centering
Worst-case BF
\end{minipage} \\
\midrule\noalign{}
\endhead
\bottomrule\noalign{}
\endlastfoot
Coding error: \(x = 0 \,/\, 2 \,/\, 4\) observations re-coded & 29.57 /
2.64 / 0.38 & 4.00 / 0.64 / 0.24 \\
Observation bias: \(\omega = 1 \,/\, 2 \,/\, 5\) & 29.57 / 4.62 / 0.81 &
4.00 / 1.01 / 0.40 \\
Smoking gun on one pro-\(H_1\) observation: \(w = 1 \,/\, 3 \,/\, 5\) &
29.57 / 88.04 / 270 & 4.00 / 9.44 / 23.92 \\
Weights on \(\theta\), rescaled onto each half: Beta\((1,1)\) /
Beta\((1/2,1/2)\) / Beta\((10,10)\) & 29.57 / 13.72 / 298 & 4.00 / 3.02
/ 6.03 \\

\end{longtable}

How sensitive are these Bayes factors to the coding, the search, and the
weights on \(\theta\)? If two of the eight observations we read as
favoring the voters are re-read as favoring the entrepreneurs, the
uniform-weights Bayes factor falls from 29.6 to 2.6 --- below 20 --- and
the worst-case Bayes factor falls from 4.00 to 0.64. If evidence for the
voters were twice as easy to find as evidence for the entrepreneurs
(\(\omega = 2\)), the uniform-weights Bayes factor falls to 4.6. The
worst-case Bayes factor falls to 1.01: in the worst case, the evidence
is almost equally likely under both hypotheses, and the uniform-weights
Bayes factor is not large enough. The Jeffreys weights,
Beta\((1/2, 1/2)\) rescaled onto each half, give a Jeffreys-weights
Bayes factor of 13.7. Beta\((10, 10)\) rescaled onto each half, which
puts each theory's \(\theta\) near the middle of its half, gives 297.8.
The worst-case Bayes factor takes those weights on the working theory's
range only, since the rival's claim is evaluated at one half whatever
the rival's proponent states: 3.0 with the Jeffreys weights and 6.0 with
Beta\((10, 10)\). The Bayes factors move this far because few
observations carry them: with ten agreed observations, each observation
carries more of the conclusion than in a study with seventeen (Winward,
below). The 14 we set aside as ambiguous give a reader many observations
to disagree about, and the table says how much disagreement would change
each Bayes factor. That disagreement would begin with the two episodes
that fit both stories, the February 1867 ballot-tearing and the 1870
sequencing claim; to change the reports, it would have to continue into
the eight agreed-upon observations. A historian of Prussia would have to
read the observations both coders took as evidence of voters
coordinating on aggrievement, rather than recording entrepreneurs
building the coalition. Whether those observations can be read that way
is a question about Prussia, not about our model --- and the historians
who know these sources can answer it.

\subsection{Steinsson (2024): a single case in which no agreed
observation favors the rivals}\label{sec-app-steinsson}

Steinsson (2024) presents a single-case process trace of English
Wikipedia, asking why the encyclopedia transformed from a venue that
lent credence to fringe content in its early years into one that
proactively debunks it (p.235, abstract). Steinsson's working theory
(\(H_1\)) traces the change to internal dynamics among Wikipedia
editors: an Anti-Fringe (AF) editor camp won early disputes over the
Neutral Point of View rule, gained institutional power through
experience, numerical advantage, and a sourcing hierarchy, and the
opposing Pro-Fringe (PF) camp gradually exited the platform. The shift
in content followed from the shift in editor population. The four
explicit rival explanations are external events such as Trump and ``fake
news''; an influx of new editors with new ideologies; slow attitudinal
change among Wikipedians; and changes in the news and scientific sources
that Wikipedia cites (pp.~249--250). They take one label \(H_R\) because
each explains the change by something other than the editor turnover
that the working theory names, which is one claim about \(\theta\): at
most half of the evidence supports the turnover account. Our counts ask
of each observation only whether it supports the working theory, so a
rival label can cover any number of alternatives that agree on that
much, and Steinsson's four never have to be told apart here because no
agreed observation favors any of them.

The evidence here keeps growing: Wikipedia produces its own record, and
disputes, arbitration rulings, and editor exits keep accumulating after
any study of the case stops. Yet everything in it is written down, and
each observation we coded can be checked. We report both Bayes factors,
as we do for every study in this section, and for this study, the two
agree about the conclusion.

We coded the article into 23 observations, recorded in the coding table
at \texttt{replications/steinsson\_2023/cases.csv}, with one observation
per line, including a page citation and a verbatim quote. Both coders
agreed that 12 of the 23 favor \(H_1\). They agreed on no observation
favoring the rivals --- \(r = 0\). We set aside the remaining 11 as
ambiguous. The 12 agreed observations include positive evidence for the
editor-population-shift theory --- AF victories in the early arbitration
cases (p.246), AF-aligned authorship of three supporting guidelines
(p.246), the source-deprecation hierarchy starting with the Daily Mail
in 2017 (p.247), the disproportionate exits of PF voters across five
hotly contested referenda (p.249, Table 3), the documented content shift
across 63 articles (p.241, Table 1) --- together with Steinsson's tests
of those four explicit alternative explanations, each of which Steinsson
finds inconsistent with the data (pp.249--250). The 11 ambiguous
observations include several that the skeptical coder flagged as
potentially favoring the rivals --- the
early-Wikipedians-were-scientists footnote on p.250, the PF-editor
administrative-bias dossier in footnote 27 on p.248, Sanger's later
condemnation of NPOV interpretations on p.240, the Croatian-Wikipedia
counterexample on p.236 --- but the charitable coder conceded none of
them. An observation flagged by a single coder counts as ambiguous, not
as evidence for the rivals.

At \((k, r) = (12, 0)\) the uniform-weights Bayes factor is 8,191 and
the worst-case Bayes factor is 630.1 --- both far above the BF
\(\geq 20\) threshold, and the worst-case value carries the proved error
rate with no claim about the size of the evidence base
(Theorem~\ref{thm-error-rate}): if we set aside \(H_R\) whenever the
worst-case BF is above the one computed here, we would wrongly reject it
less than 0.01\% of the time. Unanimous evidence is strong under both
Bayes factors. The size of any one number is not the finding. The
questions that matter are the ones the sensitivity table answers: how
many of the twelve agreed observations would have to be re-read as
evidence for the rivals, how much easier pro-\(H_1\) evidence would have
to have been to find than pro-rival evidence, how much weight one
observation would have to carry, and how the weights on \(\theta\) would
have to change, before the Bayes factors fall below 20.

\begin{longtable}[]{@{}
  >{\raggedright\arraybackslash}p{(\linewidth - 4\tabcolsep) * \real{0.6619}}
  >{\centering\arraybackslash}p{(\linewidth - 4\tabcolsep) * \real{0.1871}}
  >{\centering\arraybackslash}p{(\linewidth - 4\tabcolsep) * \real{0.1511}}@{}}

\caption{\label{tbl-app-steinsson-sens}Bayes factors for Steinsson
(2024) when coding error, observation bias, a smoking-gun weight, and
the weights on \(\theta\) each take three values. Agreed counts
\((k, r) = (12, 0)\); 11 observations set aside as ambiguous. Coding
error takes larger values (\(x = 0/4/8\)) because both Bayes factors are
far above 20 at the agreed coding.}

\tabularnewline

\toprule\noalign{}
\begin{minipage}[b]{\linewidth}\raggedright
What we vary
\end{minipage} & \begin{minipage}[b]{\linewidth}\centering
Uniform-weights BF
\end{minipage} & \begin{minipage}[b]{\linewidth}\centering
Worst-case BF
\end{minipage} \\
\midrule\noalign{}
\endhead
\bottomrule\noalign{}
\endlastfoot
Coding error: \(x = 0 \,/\, 4 \,/\, 8\) observations re-coded & 8,191 /
6.49 / 0.15 & 630 / 1.10 / 0.17 \\
Observation bias: \(\omega = 1 \,/\, 2 \,/\, 5\) & 8,191 / 325 / 24.04 &
630 / 35.03 / 4.39 \\
Smoking gun on one pro-\(H_1\) observation: \(w = 1 \,/\, 3 \,/\, 5\) &
8,191 / 32,767 / 131,071 & 630 / 2,184 / 7,710 \\
Weights on \(\theta\), rescaled onto each half: Beta\((1,1)\) /
Beta\((1/2,1/2)\) / Beta\((10,10)\) & 8,191 / 6,219 / 86,511 & 630 / 957
/ 180 \\

\end{longtable}

Only the coding question brings the uniform-weights Bayes factor and the
worst-case Bayes factor below 1. If four of the twelve agreed
observations are re-read as favoring the rivals, the uniform-weights
Bayes factor falls to 6.49 and the worst-case Bayes factor to 1.10. If
eight of the twelve favored the rivals, the uniform-weights Bayes factor
is 0.15 and the worst-case value 0.17 --- the worst-case Bayes factor is
defined at every count, including counts favoring the rivals, where it
falls below one. No other question in the table brings the
uniform-weights Bayes factor below 20 or the worst-case Bayes factor
below 1 on its own: a search twice as likely to find pro-\(H_1\)
evidence as pro-rival evidence, whatever the underlying evidence base
contains, leaves the uniform-weights Bayes factor at 325 and the
worst-case Bayes factor at 35.03, and a search five times as likely
leaves them at 24.04 and 4.39, respectively. Stating other weights on
\(\theta\) replaces it with a different Bayes factor, and both are above
20: the Jeffreys-weights Bayes factor is 6,219, and Beta\((10, 10)\)
rescaled onto each half gives 86,511. The worst-case Bayes factor is 957
with the Jeffreys weights and 180 with Beta\((10, 10)\). Treating one
pro-\(H_1\) observation as a smoking gun only raises both Bayes factors.

A reader who doubts this conclusion has one route: the coding. To bring
the uniform-weights Bayes factor below 20, a reader must re-read at
least four of the twelve agreed observations as evidence for the rivals.
The eleven ambiguous observations, each documented with its verbatim
quote and the reasoning of the coder who flagged it, are where such a
re-reading would start; to change the conclusion, it would have to
continue into the agreed twelve. Whether the early Wikipedians'
scientific backgrounds, the administrative-bias dossier, or the Croatian
counterexample favor the rivals --- and whether any of the agreed twelve
should be re-coded --- are questions about Wikipedia's recorded history,
not about our model, and the editors and scholars who know that history
can answer them. The framework's contribution here is not the verdict
--- Steinsson could compute the numbers. It is the table, which says how
much rests on each part of the inference: nearly everything on the
coding of the twelve agreed observations, and almost nothing on the
weights on \(\theta\); the search would have to have favored pro-\(H_1\)
evidence more than five to one before the uniform-weights Bayes factor
came near 20.

\subsection{Andersen (2024): a three-case comparison where the two Bayes
factors fall on opposite sides of 20}\label{sec-app-andersen}

Andersen (2024) asks why Scandinavia achieved extensive, peaceful
agrarian reforms that underwrote stable democratization, whereas France
and Prussia did not. Andersen's working theory (\(H_1\)) is that
meritocratic recruitment to the central administration, plus state
control over local administration, produced impartial state-society
relations and smooth reform. The rival (\(H_R\)) is the Moore-style
socioeconomic-structure account: peasant rebellions and elite violence
drive reforms, and state capacity is a consequence rather than a cause.
The design is a three-case comparative process trace.

We coded the article into 25 observations, recorded in the coding table
at \texttt{replications/andersen\_2023/cases.csv}, with one observation
per line, including a page citation and a verbatim quote. Both coders
agreed that 9 observations favor \(H_1\) and that 3 favor \(H_R\). The 9
pro-\(H_1\) observations include the cross-case productivity test that
rules out a Moore-style productivity rival (p.59), the 1719 Swedish
merit-based public-employment rules (p.62), the Danish 1736 mandate that
public employees hold a Copenhagen law degree (p.62), the abolition of
Stavnsb\aa{}ndet in 1788 and the peaceful land redistribution that
followed (p.58), and the negative-case evidence from France and Prussia
(pp.58--59). The 3 pro-\(H_R\) observations are concessions Andersen
makes to the rival: the violent Dalecarlian Rebellion of 1743 (p.58),
the 1762 Bergen tax riot specifically about corrupt distribution of
trading rights (p.58), and the peasant estate's staunch support for
limited suffrage in Sweden in 1866 (p.56). We set the remaining 13
observations aside as ambiguous. Eleven of the 13 were flagged by only
one coder --- seven by the skeptical coder alone, among them Gustav
III's 1792 assassination, the 1811 Kl\aa{}gerup gathering, Norway's 1814
odelsretten reversal, and the concession about suffrage violence in
Sweden in 1917 --- and under the merge rule, a flag from one coder alone
enters neither count.

At \((k, r) = (9, 3)\) --- the same counts as those in the paper's
running example --- the uniform-weights Bayes factor is 20.67, slightly
above the BF \(\geq 20\) threshold, while the worst-case Bayes factor is
2.73, above 1. The evidence supports a conclusion in favor of \(H_1\),
even under the worst-case scenario, although the margin is narrow. How
robust are these results to different specifications?

\begin{longtable}[]{@{}
  >{\raggedright\arraybackslash}p{(\linewidth - 4\tabcolsep) * \real{0.6866}}
  >{\centering\arraybackslash}p{(\linewidth - 4\tabcolsep) * \real{0.1567}}
  >{\centering\arraybackslash}p{(\linewidth - 4\tabcolsep) * \real{0.1567}}@{}}

\caption{\label{tbl-app-andersen-sens}Bayes factors for Andersen (2024)
when coding error, observation bias, a smoking-gun weight, and the
weights on \(\theta\) each take three values. Agreed counts
\((k, r) = (9, 3)\); 13 observations set aside as ambiguous.}

\tabularnewline

\toprule\noalign{}
\begin{minipage}[b]{\linewidth}\raggedright
What we vary
\end{minipage} & \begin{minipage}[b]{\linewidth}\centering
Uniform-weights BF
\end{minipage} & \begin{minipage}[b]{\linewidth}\centering
Worst-case BF
\end{minipage} \\
\midrule\noalign{}
\endhead
\bottomrule\noalign{}
\endlastfoot
Coding error: \(x = 0 \,/\, 2 \,/\, 4\) observations re-coded & 20.67 /
2.44 / 0.41 & 2.73 / 0.56 / 0.23 \\
Observation bias: \(\omega = 1 \,/\, 2 \,/\, 5\) & 20.67 / 2.85 / 0.40 &
2.73 / 0.67 / 0.21 \\
Smoking gun on one pro-\(H_1\) observation: \(w = 1 \,/\, 3 \,/\, 5\) &
20.67 / 55.89 / 156 & 2.73 / 5.90 / 13.68 \\
Weights on \(\theta\), rescaled onto each half: Beta\((1,1)\) /
Beta\((1/2,1/2)\) / Beta\((10,10)\) & 20.67 / 9.53 / 253 & 2.73 / 2.07 /
4.38 \\

\end{longtable}

Under both Bayes factors, Andersen's conclusion is sensitive to the
coding, the search, and the weights on \(\theta\) alike. If two of the
nine observations both coders agreed favor \(H_1\) are re-read as
favoring the rival, the uniform-weights Bayes factor falls to 2.4, and
the worst-case Bayes factor falls below 1. If pro-\(H_1\) evidence were
twice as easy to find in these histories as pro-rival evidence
(\(\omega = 2\)), they fall to 2.9 and 0.67, respectively. The Jeffreys
weights, Beta\((1/2, 1/2)\) rescaled onto each half, give a
Jeffreys-weights Bayes factor of 9.5, the value the applications table
reports. The re-reading and the search each take the uniform-weights
Bayes factor below 20 and the worst-case Bayes factor below 1 on their
own, and the Jeffreys-weights Bayes factor is below 20 as well. In the
full table, two changes leave a Bayes factor above 20: the smoking-gun
weight, and Beta\((10, 10)\) rescaled onto each half, under which each
theory predicts a value of \(\theta\) near the middle of its half, and
the Bayes factor those weights give is 252.9. Andersen's evidence
supports \(H_1\), but it is sensitive to different specifications. Our
model cannot say whether some piece of evidence has been miscoded or
whether the research favored pro-\(H_1\) evidence. Those are questions
about Scandinavian historiography, not about the arithmetic. Historians
of Scandinavia can answer them. The table cannot.

\subsection{Hammoud-Gallego and Freier (2023): symbolic refugee
protection in a mixed-methods
design}\label{sec-app-hammoud-gallego-freier}

Hammoud-Gallego and Freier (2023) ask why Latin American states
liberalized refugee legislation in the early twenty-first century, when
conventional determinants of immigration and refugee policy ---
immigrant and refugee stocks, emigrant numbers, democratization --- do
not predict liberalization in the region. Their working theory (\(H_1\))
is \emph{symbolic refugee protection}: leftist Pink-Tide ideology,
regional integration, and human-rights signaling drove liberalization.
The rival (\(H_R\)) in our coding is the conventional, instrumental
story: liberalization responded to migration pressures at home and
foreign-policy pressures from abroad, not to symbolic politics.
Hammoud-Gallego and Freier themselves formulate six hypotheses rather
than two; our two-theory coding sets their symbolic account against the
conventional determinants. The paper combines large-N quantitative
analysis with within-case process tracing of Argentina and Mexico, and
we coded findings from both as observations.

We coded 28 observations from the paper into the coding table in
\texttt{replications/hammoud\_gallego\_freier\_2022/cases.csv}, one
observation per line, with a page citation and a verbatim quote. Both
coders agreed that 11 favor \(H_1\) and 2 favor \(H_R\). The 11
pro-\(H_1\) observations include the leftist-government coefficient in
the Tobit specification (p.459), the trade-openness coefficient (p.459),
the spatial autocorrelation coefficients indicating regional diffusion
(p.460), the chronological correlation of human-rights treaties with
regulatory complexity in Argentina and Mexico (p.464), and Argentine
interviews on the technical rather than public character of the refugee
law (p.462). They also include the two null results, on immigrant and
refugee stocks and on emigrant numbers: Hammoud-Gallego and Freier put
both among the predictors the prior literature names and their own
models do not support, and explain liberalization by leftist ideology
and human-rights signaling rather than by either, so a null on either is
a null on a conventional determinant, which is the rival in our
coding.\footnote{Our earlier coding credited the emigrant-numbers null
  to the rival, on the ground that diaspora politics is one of the
  working theory's own mechanisms. The article's abstract says the
  opposite --- its models ``do not support some consistent predictors of
  policy liberalization identified by the literature such as immigrant
  and refugee stocks, democratization, and the number of emigrants'' ---
  and the two nulls come from one sentence of the source, so we now code
  them alike.} The 2 pro-\(H_R\) observations are the Gonz\'alez-Murphy
account of US bilateral pressure on Mexico (``a slap in the face with a
white glove,'' p.466) and the authors' own concession that three of six
theoretical hypotheses receive no statistical support (p.468). We set
aside the remaining 15 observations as ambiguous, mostly because the
published findings discriminate poorly between symbolic adoption and the
instrumental story: the 2010 Tamaulipas massacre read as a
critical-juncture trigger, the UNHCR's facilitating role, and several
interview testimonies that the charitable coder counted as favoring
\(H_1\) but the skeptical coder did not concede.

At \((k, r) = (11, 2)\) the uniform-weights Bayes factor is 154, above
the BF \(\geq 20\) threshold. The worst-case Bayes factor is 14.91, well
above 1. The ambiguous count is itself a substantive finding: our coders
could not agree on fifteen of the twenty-eight observations --- more
than in any other study we coded --- because the published evidence
often fits the symbolic and instrumental stories equally well. Seven of
its thirteen agreed observations are results of the paper's quantitative
analysis. Without them, the counts are five-and-one, and the
uniform-weights Bayes factor is 15.0, below the threshold: that is the
number for a reader who counts the process-tracing evidence alone.

\begin{longtable}[]{@{}
  >{\raggedright\arraybackslash}p{(\linewidth - 4\tabcolsep) * \real{0.6765}}
  >{\centering\arraybackslash}p{(\linewidth - 4\tabcolsep) * \real{0.1544}}
  >{\centering\arraybackslash}p{(\linewidth - 4\tabcolsep) * \real{0.1691}}@{}}

\caption{\label{tbl-app-hgf-sens}Bayes factors for Hammoud-Gallego and
Freier (2023) when coding error, observation bias, a smoking-gun weight,
and the weights on \(\theta\) each take three values. Agreed counts
\((k, r) = (11, 2)\); 15 observations set aside as ambiguous.}

\tabularnewline

\toprule\noalign{}
\begin{minipage}[b]{\linewidth}\raggedright
What we vary
\end{minipage} & \begin{minipage}[b]{\linewidth}\centering
Uniform-weights BF
\end{minipage} & \begin{minipage}[b]{\linewidth}\centering
Worst-case BF
\end{minipage} \\
\midrule\noalign{}
\endhead
\bottomrule\noalign{}
\endlastfoot
Coding error: \(x = 0 \,/\, 2 \,/\, 4\) observations re-coded & 154 /
10.14 / 1.53 & 14.91 / 1.49 / 0.41 \\
Observation bias: \(\omega = 1 \,/\, 2 \,/\, 5\) & 154 / 12.07 / 1.45 &
14.91 / 1.87 / 0.51 \\
Smoking gun on one pro-\(H_1\) observation: \(w = 1 \,/\, 3 \,/\, 5\) &
154 / 477 / 1,523 & 14.91 / 38.93 / 107 \\
Weights on \(\theta\), rescaled onto each half: Beta\((1,1)\) /
Beta\((1/2,1/2)\) / Beta\((10,10)\) & 154 / 60.32 / 3,917 & 14.91 /
11.04 / 21.15 \\

\end{longtable}

This conclusion too is sensitive to the coding, the search, and the
weights on \(\theta\). If two of the eleven observations both coders
agreed favor \(H_1\) are re-read as favoring the rival, the
uniform-weights Bayes factor falls to 10.1 and the worst-case value to
1.5. If evidence for symbolic protection were twice as easy to find in
the published sources as evidence for the instrumental story
(\(\omega = 2\)), the uniform-weights Bayes factor falls to 12.1 and the
worst-case Bayes factor to 1.87. The Jeffreys weights,
Beta\((1/2, 1/2)\), rescaled onto each half, give a Jeffreys-weights
Bayes factor of 60.3, still above 20, and Beta\((10, 10)\), rescaled
onto each half, gives 3,917. The worst-case Bayes factor is 11.0 with
the Jeffreys weights and 21.1 with Beta\((10, 10)\). Whether two of
those eleven observations should be re-read, and whether a search of
these published sources could have favored the symbolic story two to
one, are questions about Latin American refugee policy, not about our
model. The fifteen ambiguous observations are documented in the coding
table, each with its verbatim quote and the noted disagreement; scholars
of Latin American migration policy can re-read them, and the table
indicates how many re-readings would change the conclusion.
Hammoud-Gallego and Freier can make the case for their interpretation,
which has a substantive argument; a reader who disagrees can name the
observations they would re-code. The conversation that follows is
quantitative rather than rhetorical.

\subsection{Pavone and Stiansen (2022): two pre-specified rivals and the
weight of one letter}\label{sec-app-pavone-stiansen}

Pavone and Stiansen (2022) process-trace Norway's 2019 social-benefits
reform. Did policy makers reform preemptively, anticipating EFTA Court
judicial review (the working theory, ``shadow effect of courts,''
\(m_2\) in their notation)? Or did they respond to managerial
recognition of legal obligation, as the literature would predict
(\(m_1\), our \(H_R\))? Pavone and Stiansen name both mechanisms in
advance and test them against the same archive --- correspondence, audit
reports, parliamentary testimony --- asking which one the timing and
substance of the reform fits. Two rivals and one shared archive: this is
the comparison the framework was built for.

We coded 22 observations from the article into the coding table in
\texttt{replications/pavone\_stiansen\_2021/cases.csv}, one observation
per line, with a page citation and a verbatim quote. Both coders agreed
that 7 favor \(H_1\) and 4 favor \(H_R\); we set aside the remaining 11
as ambiguous. Two institutions recur in the evidence: NAV, the Norwegian
welfare agency whose benefits practice was at issue, and the NIC, a
quasi-judicial body within the Norwegian state that can refer questions
to the EFTA Court. The 7 observations favoring \(H_1\) include NAV's
decade-long suppression of internal legal concerns (p.330), NAV's
treatment of the NIC as a subordinate advisory body rather than
appealing through the courts (p.331), the November 2018 NIC letter
explicitly threatening to refer to the EFTA Court (p.331), NAV's
documented motive in the January 2019 letter --- to reduce the
likelihood of EFTA Court referral, in their own words (p.332) --- and
NAV's January 2019 jurisdictional gambit to argue that the NIC lacked
standing to refer (p.332). The 4 observations favoring \(H_R\) are the
textbook \(m_1\) evidence Pavone and Stiansen themselves engage with:
their concession that initial noncompliance is partially attributable to
insufficient legal knowledge (p.330), the 2017 Tolley CJEU ruling
against the United Kingdom on point (p.330), the substance of the 2019
reform (full EEA compliance plus reopening old cases plus compensating
victims, p.329) --- what \(m_1\) predicts once legal obligation is
recognized --- and the consistent on-record managerial framing by the
Attorney General, the NAV Director, and the Director of Public
Prosecutions (p.329).

At \((k, r) = (7, 4)\) the uniform-weights Bayes factor is 4.16, and the
worst-case Bayes factor is 0.83 --- the one study in the set whose
worst-case value falls below one, because seven of eleven is close to
the even split the rival predicts. With two rivals specified in advance
and a shared documentary archive, neither mechanism is much better
supported than the other when we count only the observations both coders
agreed on. The smoking-gun observations Pavone and Stiansen rely on ---
NAV's own internal correspondence acknowledging the EFTA Court
motivation --- appear in the same archive as the textbook
managerial-framing observations: the substance of the reform and the
consistent public framing by senior officials. Coding with two coders
makes that tension explicit.

\begin{longtable}[]{@{}
  >{\raggedright\arraybackslash}p{(\linewidth - 4\tabcolsep) * \real{0.6866}}
  >{\centering\arraybackslash}p{(\linewidth - 4\tabcolsep) * \real{0.1642}}
  >{\centering\arraybackslash}p{(\linewidth - 4\tabcolsep) * \real{0.1493}}@{}}

\caption{\label{tbl-app-ps-sens}Bayes factors for Pavone and Stiansen
(2022) when coding error, observation bias, a smoking-gun weight, and
the weights on \(\theta\) each take three values. Agreed counts
\((k, r) = (7, 4)\); 11 observations set aside as ambiguous.}

\tabularnewline

\toprule\noalign{}
\begin{minipage}[b]{\linewidth}\raggedright
What we vary
\end{minipage} & \begin{minipage}[b]{\linewidth}\centering
Uniform-weights BF
\end{minipage} & \begin{minipage}[b]{\linewidth}\centering
Worst-case BF
\end{minipage} \\
\midrule\noalign{}
\endhead
\bottomrule\noalign{}
\endlastfoot
Coding error: \(x = 0 \,/\, 2 \,/\, 4\) observations re-coded & 4.16 /
0.63 / 0.08 & 0.83 / 0.29 / 0.15 \\
Observation bias: \(\omega = 1 \,/\, 2 \,/\, 5\) & 4.16 / 0.79 / 0.10 &
0.83 / 0.32 / 0.05 \\
Smoking gun on one pro-\(H_1\) observation: \(w = 1 \,/\, 3 \,/\, 5\) &
4.16 / 10.14 / 25.04 & 0.83 / 1.49 / 2.89 \\
Weights on \(\theta\), rescaled onto each half: Beta\((1,1)\) /
Beta\((1/2,1/2)\) / Beta\((10,10)\) & 4.16 / 2.69 / 16.23 & 0.83 / 0.74
/ 1.02 \\

\end{longtable}

For Pavone and Stiansen, the question runs in the opposite direction
from the other applications: not how many observations would have to be
re-read for the conclusion to fail, but what a reader would have to
grant for the Bayes factors to exceed 20. Of the four questions in the
table, only the smoking-gun weight raises the Bayes factors; re-coding
and observation bias lower both, and a rival-tilted prior lowers the
uniform-weights Bayes factor, the only one of the two whose denominator
depends on weights over the values of \(\theta\) consistent with the
rival's claim. The candidate for that weight is a single document: NAV's
January 2019 internal correspondence, in which the proposed reform's
stated motive is to preclude EFTA Court adjudication. A weight of 3 on
that one observation brings the uniform-weights Bayes factor to 10.1 and
the worst-case value to 1.49; a weight of 5 raises the uniform-weights
Bayes factor just above 20, to 25.0, and the worst-case value to 2.89;
and at the weight of 10 the main paper uses for its smoking-gun
illustration, the weighted counts are (16, 4) and the worst-case value
is 20.5. Pavone and Stiansen's argument is, in effect, that this one
piece of evidence is much more probative than the routine confirmations
of the legal-obligation story, and that the numbers cannot settle
whether it is. Is one letter, in which NAV states its motive in its own
words, worth ten routine confirmations? Scholars of courts and
compliance can argue that question; the table says the conclusion turns
on exactly that.

\subsection{Winward (2021): within-region variation in mass
violence}\label{sec-app-winward}

Winward (2021) asks why mass categorical violence took different forms
and frequencies across regions during Indonesia's 1965--66 killings. The
working theory (\(H_1\)) is that low state intelligence capacity forced
security forces to rely on civilian elites for information; civilian
elites widened targeting criteria; and logistical strain from large
detainee populations drove mass executions. The rivals (\(H_R\)) are
Balcells-style political-cleavage explanations --- violence tracks the
strength of an out-group's local opponents, not state capacity --- and
Hoover-Green militia-empowerment dynamics. Winward compares the forms
and frequencies of violence across three provinces: Central Java, East
Java, and West Java.

We coded 27 observations from Winward's article into the coding table in
\texttt{replications/winward\_2020/cases.csv}, one observation per line,
with a page citation and a verbatim quote. Both coders agreed that 14
favor \(H_1\) and 3 favor \(H_R\); we set aside the remaining 10 as
ambiguous. The 14 pro-\(H_1\) observations include the antecedent low
intelligence capacity in Central Java (p.16), the underdeveloped
sub-district command structure (p.17), the RPKAD's explicit solicitation
of intelligence from PKI rivals (p.18), the documented militia training
of 24,000 youths in Surakarta (p.18), and the Muhammadiyah denunciation
of LEKRA cultural-event attendees (p.19). They also include the
proliferation of makeshift prisons (p.20), the near-ubiquitous torture
in Central Java prisons (p.21), the US Embassy airgram explicitly tying
executions to detainee logistics (p.21), and the cross-province
killing-to-incarceration ratios: 1:1 in West Java, 2:1 in Central Java,
and 8:1 in East Java (p.12). The 3 pro-\(H_R\) observations are
Chandra's finding that violence tracks PKI opponents, not PKI size
(p.20); the Garut sub-case, in which ex-Darul-Islam militia drove mass
killing in a high-capacity province (p.24); and Winward's own concession
that the West Java versus East Java comparison is imperfect (p.11).

At \((k, r) = (14, 3)\) the uniform-weights Bayes factor is 264 and the
worst-case Bayes factor is 21.3. This is one of the two studies in the
set whose worst-case Bayes factor is above 20 on its own: a researcher
who sets aside \(H_R\) with this margin will do it wrongly less than 5\%
of the time. Seventeen agreed observations, fourteen of them favoring
\(H_1\), put the uniform-weights Bayes factor far enough above 20 that
no single observation carries it. The ten ambiguous observations are
mostly episodes below the provincial level --- Subang, Cirebon, and
Suharto's role. One coder read these episodes as instances of Winward's
province-level argument; the other read variation within provinces as
grounds for doubting coding at the province level.

\begin{longtable}[]{@{}
  >{\raggedright\arraybackslash}p{(\linewidth - 4\tabcolsep) * \real{0.6715}}
  >{\centering\arraybackslash}p{(\linewidth - 4\tabcolsep) * \real{0.1606}}
  >{\centering\arraybackslash}p{(\linewidth - 4\tabcolsep) * \real{0.1679}}@{}}

\caption{\label{tbl-app-wn-sens}Bayes factors for Winward (2021) when
coding error, observation bias, a smoking-gun weight, and the weights on
\(\theta\) each take three values. Agreed counts \((k, r) = (14, 3)\);
10 observations set aside as ambiguous.}

\tabularnewline

\toprule\noalign{}
\begin{minipage}[b]{\linewidth}\raggedright
What we vary
\end{minipage} & \begin{minipage}[b]{\linewidth}\centering
Uniform-weights BF
\end{minipage} & \begin{minipage}[b]{\linewidth}\centering
Worst-case BF
\end{minipage} \\
\midrule\noalign{}
\endhead
\bottomrule\noalign{}
\endlastfoot
Coding error: \(x = 0 \,/\, 2 \,/\, 4\) observations re-coded & 264 /
19.78 / 3.16 & 21.34 / 2.24 / 0.57 \\
Observation bias: \(\omega = 1 \,/\, 2 \,/\, 5\) & 264 / 12.12 / 1.01 &
21.34 / 1.65 / 0.37 \\
Smoking gun on one pro-\(H_1\) observation: \(w = 1 \,/\, 3 \,/\, 5\) &
264 / 775 / 2,337 & 21.34 / 54.04 / 143 \\
Weights on \(\theta\), rescaled onto each half: Beta\((1,1)\) /
Beta\((1/2,1/2)\) / Beta\((10,10)\) & 264 / 91.58 / 15,795 & 21.34 /
15.19 / 34.53 \\

\end{longtable}

What would it take to change this conclusion? Two re-readings among the
fourteen agreed pro-\(H_1\) observations reduce the uniform-weights
Bayes factor to 19.78, below the threshold, while the worst-case Bayes
factor remains above 1. In the worst case, the evidence is still more
likely under \(H_1\). However, the uniform-weights Bayes factor is not
strong enough. The uniform-weights Bayes factor falls to 12.1 if
evidence favoring the capacity theory were twice as easy to find as
evidence favoring the rivals (\(\omega = 2\)), even when the worst-case
Bayes factor drops to 1.65 and remains above 1. Both alternative
weightings leave the Bayes factor above 20: the Jeffreys-weights Bayes
factor is 91.6, and Beta\((10, 10)\) rescaled onto each half gives
15,795. The worst-case Bayes factor is 21.3 under uniform weights, 15.2
under Jeffreys weights, and 34.5 under Beta\((10, 10)\), all of which
are above 1. Strong findings benefit from this exercise as much as weak
ones do: each answer replaces an adjective with a quantity, so rather
than calling the evidence ``strong'' or ``decisive,'' Winward can say
how many re-readings or how much search imbalance would put each Bayes
factor below its respective threshold. Whether any of those fourteen
observations should be read the other way, and whether a search of these
sources could have turned up evidence for the capacity story twice as
readily as evidence for the cleavage and militia stories, are questions
about Central Java, East Java, and West Java, not about our model.
Scholars of the 1965--66 Indonesian killings can answer them, and the
table shows how many re-readings or how much imbalance it would take.

\subsection{What these applications teach us about the
framework}\label{sec-app-lessons}

Three lessons emerge from the six applications. First, the framework
applies to a wide range of process-tracing designs: within-case
(Steinsson, Pavone, and Stiansen), comparative across a small number of
cases (Andersen, Mor), regional comparison (Winward), and a
mixed-methods design (Hammoud-Gallego and Freier). The same model
handles all of them because it treats each piece of evidence
identically, regardless of the case it came from. An author trained to
think of within-case process tracing and small-N comparison as
statistically distinct can use the framework as one common tool for
both, and a reader of two papers in different traditions can compare
them on a common scale.

Second, the sensitivity table changes what an author and a reviewer can
say to each other. Without it, a reviewer who doubts a published claim
must either accept the author's narrative or reject it. With it, the
reviewer can ask questions that have specific answers --- how many
observations would have to be re-coded, how much easier pro-\(H_1\)
evidence would have to have been to find than pro-rival evidence, how
much weight one observation would have to carry, how far toward the
rival a prior would have to tilt --- and the table answers each one. The
answers differ across the six studies. For Winward, re-coding two of the
fourteen observations that both coders agreed to favor of \(H_1\) leaves
the uniform-weights Bayes factor at 19.8, just below 20. For Andersen,
the same two re-codings yield a uniform-weights Bayes factor of 2.4. For
Pavone and Stiansen's two-rival design, both Bayes factors are below 20
at the agreed coding; the uniform-weights Bayes factor exceeds 20 only
when NAV's January 2019 letter carries the weight of five observations,
and the worst-case value only at a weight of ten, as a sensitivity
analysis. None of these is a verdict on the paper. Each tells scholars
who know the case exactly what it would take --- how many re-coded
observations, how one-sided a search --- for the conclusion to change.

Third, the framework does not adjudicate the methodological choices that
go into it; it makes them visible. The choice between the
uniform-weights and the worst-case Bayes factor, the choice of
\(\geq 20\) as a threshold for the former and 1 for the latter, the
choice of which observations support which theory, and the choice of
which observations deserve smoking-gun weights are all live questions on
which careful readers will disagree. Each of the six subsections shows
one or more of these choices doing real work: the choice of Bayes factor
in Andersen, where the uniform-weights Bayes factor sits just above the
threshold and the worst-case Bayes factor says the counts cannot reach
it at the rival's best case; the decisions about categorizing
observations as supporting one or the other theory in Steinsson and in
Hammoud-Gallego and Freier, where observations flagged by only one coder
stay out of the agreed counts, and where counting the results of a
study's own quantitative analysis as observations is what keeps
Hammoud-Gallego and Freier and Mor above the threshold under the uniform
weights; and the smoking-gun weight in Pavone and Stiansen, where one
letter must carry the weight of five observations for the
uniform-weights Bayes factor to exceed 20. The contribution is to
present these choices in a form that enables the author and the reviewer
to collaborate productively, rather than leaving them implicit.

\section{Reproducing the running example}\label{sec-reproduction}

This section shows how to reproduce every number reported by the running
example: the uniform-weights Bayes factor, the worst-case Bayes factor,
and the sensitivity values. The functions that compute them are in the R
package \texttt{DrWrinch}, named after Dorothy Maud Wrinch (1894--1976),
the mathematician whose joint papers with Harold Jeffreys in the early
1920s developed the framework that later became Jeffreys's (1961) theory
of Bayes factors. The function \texttt{bf\_uniform\_weights} computes
the uniform-weights Bayes factor, and \texttt{bf\_worst\_case} computes
the worst-case Bayes factor. The package uses \texttt{y\_W} for the
count of pro-\(H_1\) observations and \texttt{y\_R} for the count of
pro-\(H_R\) observations, corresponding to \(k\) and \(r\) in this
supplement's notation.

The package is on GitHub at
\url{https://github.com/bowers-illinois-edu/DrWrinch} and installs with
\texttt{remotes}:

\begin{Shaded}
\begin{Highlighting}[]
\CommentTok{\# install.packages("remotes")}
\NormalTok{remotes}\SpecialCharTok{::}\FunctionTok{install\_github}\NormalTok{(}\StringTok{"bowers{-}illinois{-}edu/DrWrinch"}\NormalTok{)}
\end{Highlighting}
\end{Shaded}

\begin{Shaded}
\begin{Highlighting}[]
\FunctionTok{library}\NormalTok{(DrWrinch)}
\end{Highlighting}
\end{Shaded}

The country \countrylabel example has nine observations favoring the
working theory and three favoring the rival. The uniform-weights Bayes
factor reads \texttt{y\_W} and \texttt{y\_R} as Bernoulli successes and
failures with an unknown probability \(\theta\) of supporting \(H_1\):

\begin{Shaded}
\begin{Highlighting}[]
\FunctionTok{bf\_uniform\_weights}\NormalTok{(}\AttributeTok{y\_W =} \DecValTok{9}\NormalTok{, }\AttributeTok{y\_R =} \DecValTok{3}\NormalTok{)}
\end{Highlighting}
\end{Shaded}

\begin{verbatim}
[1] 20.67196
\end{verbatim}

This matches the uniform-weights Bayes factor reported in the main
paper. The worst-case Bayes factor for the same counts is the
closed-form expression in Section~\ref{sec-closed-forms}.

\begin{Shaded}
\begin{Highlighting}[]
\FunctionTok{bf\_worst\_case}\NormalTok{(}\AttributeTok{y\_W =} \DecValTok{9}\NormalTok{, }\AttributeTok{y\_R =} \DecValTok{3}\NormalTok{)}
\end{Highlighting}
\end{Shaded}

\begin{verbatim}
[1] 2.732168
\end{verbatim}

The observation-bias helper \texttt{sens\_bias()} reports the bias at
which the uniform-weights Bayes factor falls to 20, the bias at which
the worst-case Bayes factor falls to 1, and the smaller of the two,
\(\omega^\star\).

\begin{Shaded}
\begin{Highlighting}[]
\NormalTok{bias }\OtherTok{\textless{}{-}} \FunctionTok{sens\_bias}\NormalTok{(}\AttributeTok{y\_W =} \DecValTok{9}\NormalTok{, }\AttributeTok{y\_R =} \DecValTok{3}\NormalTok{)}
\NormalTok{bias}\SpecialCharTok{$}\NormalTok{omega\_uniform}
\end{Highlighting}
\end{Shaded}

\begin{verbatim}
[1] 1.009794
\end{verbatim}

\begin{Shaded}
\begin{Highlighting}[]
\NormalTok{bias}\SpecialCharTok{$}\NormalTok{omega\_worst}
\end{Highlighting}
\end{Shaded}

\begin{verbatim}
[1] 1.559262
\end{verbatim}

\begin{Shaded}
\begin{Highlighting}[]
\NormalTok{bias}\SpecialCharTok{$}\NormalTok{omega\_star}
\end{Highlighting}
\end{Shaded}

\begin{verbatim}
[1] 1.009794
\end{verbatim}

The coding-error helper \texttt{sens\_coding()} reports how many
pro-\(H_1\) observations would have to be re-coded as pro-rival before
the uniform-weights Bayes factor falls below 20, how many before the
worst-case Bayes factor falls to 1, and the smaller of the two:

\begin{Shaded}
\begin{Highlighting}[]
\NormalTok{coding }\OtherTok{\textless{}{-}} \FunctionTok{sens\_coding}\NormalTok{(}\AttributeTok{y\_W =} \DecValTok{9}\NormalTok{, }\AttributeTok{y\_R =} \DecValTok{3}\NormalTok{)}
\NormalTok{coding}\SpecialCharTok{$}\NormalTok{x\_uniform}
\end{Highlighting}
\end{Shaded}

\begin{verbatim}
[1] 1
\end{verbatim}

\begin{Shaded}
\begin{Highlighting}[]
\NormalTok{coding}\SpecialCharTok{$}\NormalTok{x\_worst}
\end{Highlighting}
\end{Shaded}

\begin{verbatim}
[1] 2
\end{verbatim}

\begin{Shaded}
\begin{Highlighting}[]
\NormalTok{coding}\SpecialCharTok{$}\NormalTok{x\_star}
\end{Highlighting}
\end{Shaded}

\begin{verbatim}
[1] 1
\end{verbatim}

One re-coding takes the uniform-weights Bayes factor below 20, matching
Section~\ref{sec-sens-coding}. The worst-case Bayes factor falls below 1
only after two re-codings, so the tipping point is one re-coding.

Weighted analyses --- in which one observation carries more probative
force than the others --- are computed by passing the summed integer
weights as the totals. For the smoking-gun example, the researcher
upgrades one pro-\(H_1\) observation to weight \(10\), yielding
\(W = 8 + 10 = 18\) and \(R = 3\):

\begin{Shaded}
\begin{Highlighting}[]
\FunctionTok{bf\_uniform\_weights}\NormalTok{(}\AttributeTok{y\_W =} \DecValTok{18}\NormalTok{, }\AttributeTok{y\_R =} \DecValTok{3}\NormalTok{)}
\end{Highlighting}
\end{Shaded}

\begin{verbatim}
[1] 2336.962
\end{verbatim}

\begin{Shaded}
\begin{Highlighting}[]
\FunctionTok{bf\_worst\_case}\NormalTok{(}\AttributeTok{y\_W =} \DecValTok{18}\NormalTok{, }\AttributeTok{y\_R =} \DecValTok{3}\NormalTok{)}
\end{Highlighting}
\end{Shaded}

\begin{verbatim}
[1] 143.2847
\end{verbatim}

The package documentation, accessible via
\texttt{?bf\_uniform\_weights}, \texttt{?bf\_worst\_case},
\texttt{?sens\_bias}, and \texttt{?sens\_coding}, describes each
argument. The six applications are reproduced by the scripts in
\texttt{replications/}, one directory per study;
\texttt{make\ replications} re-runs all six, and each
\texttt{results.csv} carries the counts, the uniform-weights Bayes
factor (column \texttt{binom\_bf}) and the worst-case Bayes factor.

\section*{References}\label{references}
\addcontentsline{toc}{section}{References}

\protect\phantomsection\label{refs}

\end{document}